\documentclass[conference]{IEEEtran}
\IEEEoverridecommandlockouts
\usepackage{cite}
\usepackage{xcolor}
\usepackage{graphicx}
\usepackage{subfig}
\usepackage{multicol,multirow}
\usepackage{caption}
\usepackage[linesnumbered,ruled,vlined]{algorithm2e}
\usepackage{mathtools,amssymb,amsthm,amsmath}
\usepackage{graphbox}
\usepackage{enumitem}
\usepackage{marvosym}
\def\BibTeX{{\rm B\kern-.05em{\sc i\kern-.025em b}\kern-.08em
   T\kern-.1667em\lower.7ex\hbox{E}\kern-.125emX}}
\newtheorem{lemma}{Lemma}
\newtheorem{theorem}{{\textsc{Theorem}}}
\newtheorem{definition}{\textbf{Definition}}
\newtheorem{property}{Property}[definition]
\newtheorem{example}{\textsc{Example}}
\newtheorem*{problem}{\textsc{Problem}}

\DeclareMathOperator{\degree}{deg}
\DeclareMathOperator{\support}{sup}
\DeclareMathOperator{\mts}{mts}
\DeclareMathOperator{\kspan}{spn}
\DeclareMathOperator{\krank}{rnk}
\DeclareMathOperator{\trn}{trn}
\SetKwFunction{dch}{decomph}
\SetKwFunction{dcv}{decompv}

\begin{document}
\title{Querying Cohesive Subgraph regarding Span-Constrained Triangles on Temporal Graphs with Dynamic Index Maintenance}

\author{
	\IEEEauthorblockN{
		Chuhan Hu\IEEEauthorrefmark{1}, 
		Ming Zhong\IEEEauthorrefmark{2},
            and Lei Li\IEEEauthorrefmark{3}} 
	\IEEEauthorblockA{\IEEEauthorrefmark{1}DSA Thrust, The Hong Kong University of Science and Technology (Guangzhou)}
	\IEEEauthorblockA{\IEEEauthorrefmark{2}School of Computer Science, Wuhan University, Wuhan, China}
	\IEEEauthorblockA{\IEEEauthorrefmark{3}The Hong Kong University of Science and Technology}
        chu083@connect.hkust-gz.edu.cn, clock@whu.edu.cn, thorli@ust.hk
        
}

\maketitle

\begin{abstract}
Recent advances in temporal graph research have redefined traditional static graph concepts such as triangles, motifs, and $k$-cores. Inspired by this, we introduce a novel $(k,\delta)$-truss for temporal graphs, requiring triangles to exist within sufficiently short time windows. The $(k,\delta)$-truss ensures both static and temporal cohesion, while the original $k$-truss is a special case when $\delta = \infty$. To address $(k,\delta)$-truss queries, we propose index-free and index-based approaches. Utilizing the dual containment relation of $(k,\delta)$-trusses, our indexes losslessly compress all $(k,\delta)$-trusses into map or tree structures, significantly reducing space while enabling optimal-time retrieval. To scale to large temporal graphs, we develop two index construction algorithms based on truss decomposition and truss maintenance, respectively, which substantially reduce redundant computations. \textcolor{black}{Moreover, we present techniques for the dynamic maintenance of the proposed indexes.} The experimental results demonstrate that index-based approaches process queries in interactive time and outperform the index-free approach by 2$\sim$4 orders of magnitude, while the indexes achieve compression ratios of up to $10^{-4}$ \textcolor{black}{and can be updated efficiently without rebuilding from scratch}.
\end{abstract}


\section{Introduction}

Recently, temporal graphs in which each edge is associated with a set of timestamps have drawn intensive research interests, as introduced by~\cite{masuda2016temporalnetwork}~\cite{petter2012temporalnetwork}. The typical examples of temporal graph are such as social networks~\cite{kossi2006socialnetwork}, transaction networks~\cite{huang2022transactionnetwork}, transportation networks~\cite{filip2023transponetwork}, communication networks~\cite{hidal2008comnetwork}, power networks~\cite{simeu2021powernetwork}, disease transmission networks~\cite{masuda2014diseasenetwork}, etc. In those graphs, the temporality enables a variety of time-relevant constraints in analytics, such as time order, time window, time span, etc.

\begin{figure}[t!]
    \centering
    \includegraphics[width=0.9\linewidth]{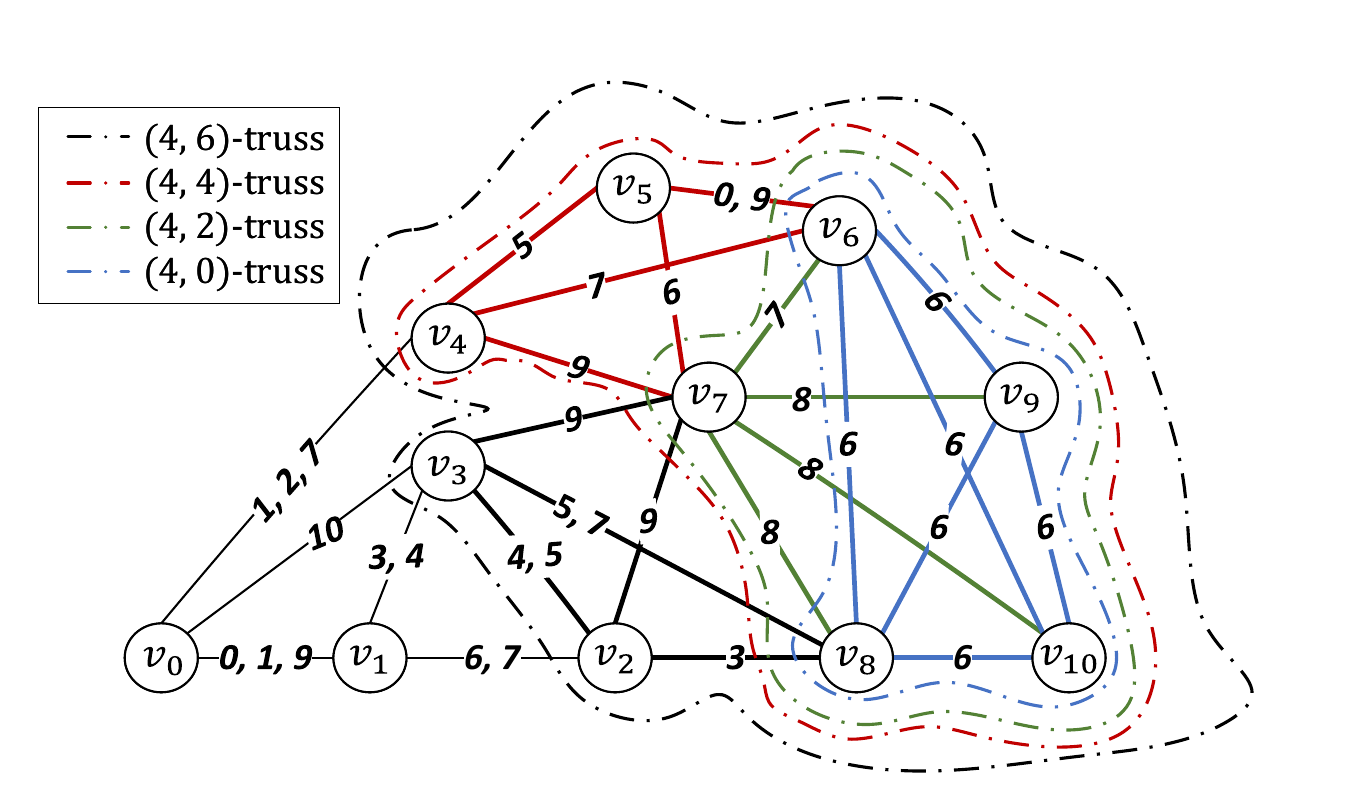}
    \caption{A running example temporal graph and several $(k,\delta)$-trusses, which forms a ``$\delta$-shell'' with a fixed $k$.}
    \label{fig:etg}
    \vspace{-0.5cm}
\end{figure}

For temporal graphs, the traditional definitions of cohesive subgraphs such as $k$-clique, $k$-truss, and $k$-core also need to be extended with time-relevant constraints, so that both temporal and topological features can be exploited for comprehensive analysis. For example, there have been a bunch of temporal $k$-core studies, which mainly fall into two categories. The first is to find primitive $k$-cores that exist in specific time windows, such as span-core~\cite{galim2018spancore}, temporal $k$-core~\cite{wu2015temporalcore}~\cite{Yang2023temporalcore}, historical $k$-core~\cite{myu2021historicalcore}~\cite{Wang2025historicalcore}, temporal $(k,\mathcal{X})$-core~\cite{zhong2024kxcore}, etc. The second is to study various new temporal $k$-cores models, such as frequent core~\cite{Bai2020frequentcore, Du2025frequentcore}, persistent core~\cite{lirh2018persistcommunity}, bursting core~\cite{chu2019burstingsubgraph}, periodic core~\cite{qin20202periodiccore}, continual core~\cite{liliu2021continuecore}, reliable core~\cite{tang2022reliablecore}, etc.

However, the models and approaches designed for $k$-core query are not sufficient for $k$-truss query on temporal graphs. In contrast to migrating them to $k$-truss directly, like~\cite{lotito2020spantruss} for the first category and~\cite{lantian2020kcommunity} for the second category, it is more important to investigate the models and approaches dedicated to temporal $k$-truss query. As we know, $k$-truss in which each edge is contained by at least $k-2$ triangles is defined on top of another more fundamental concept, namely, triangle. Thus, a reasonable definition of temporal $k$-truss should take the temporal constraint on triangles into consideration.

\begin{figure*}
    \begin{minipage}{0.65\linewidth}
    \centering
    \subfloat[$(16,\infty)$]{
        \includegraphics[width=0.28\linewidth]{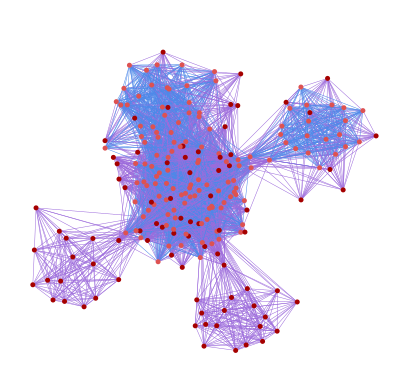}
    }
    \subfloat[$(16,200)$]{
        \includegraphics[width=0.22\linewidth]{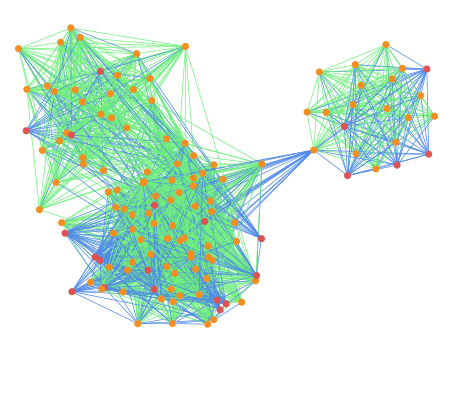}
    }
    \subfloat[$(16,150)$]{
        \includegraphics[width=0.18\linewidth]{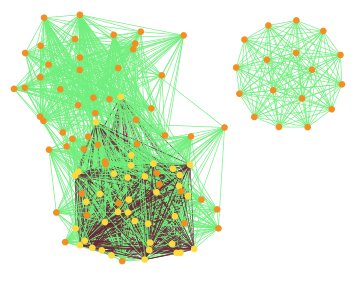}
    }
    \subfloat[$(16,100)$]{
        \includegraphics[width=0.11\linewidth]{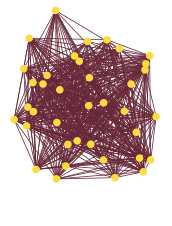}
    }
    \end{minipage}\hfill
    \begin{minipage}{0.35\linewidth}
        \centering
        {\footnotesize
        \begin{tabular}{|c|r|r|r|r|}
            \hline
             $(k,\delta)$-truss & (a) & (b) & (c) & (d)   \\
            \hline
             vertex\#& 213 & 130 & 108 & 38\\
             edge\#& 4402 & 2355 & 1735 & 564\\
             triangle\#& 42683 & 21738 & 14978 & 4670 \\
             coefficient& 0.72 & 0.77 & 0.81 & 0.85\\
             duration&803 & 803 & 803 & 802\\
             \hline
        \end{tabular}}
    \end{minipage}
    \caption{A case study on $(k,\delta)$-trusses of Email dataset, which demonstrates the effectiveness of $\delta$ on improving cohesion. In each truss, we remark the vertices and edges of successive truss with different colors, so that the changes can be observed.}
    \label{fig:case}
    \vspace{-0.5cm}
\end{figure*}

Actually, many meaningful temporal triangle or motif (note that, triangles can be seen as a kind of motifs sometimes) models~\cite{kovan2011motif, paran2017motif, liu2019motif, wang2020motif, phasa2021triangle, liu2023motif} have been proposed recently. These models mainly consider two kinds of temporal constraints. The first is the total or partial time order among edges, which is usually defined on directed temporal graphs, so that a triangle or motif can be seen as a sequence of edges in order of timestamps. Such constraints are too specific and complicated for a basic component of $k$-truss, and would result in over-tailoring of $k$-truss. The second is the duration, which is generally measured by the maximum time lag between timestamps of edges.

In this paper, we propose a novel definition of $(k,\delta)$-truss for undirected temporal graphs, on top of a kind of ``span''-constrained $\delta$-triangles. The rationale is two-fold. Since $k$-truss considers a triangle as a strong evidence of its vertices are bonded tightly in a community, it is better to consistently guarantee the tightness of triangle but not truss from the perspective of time. Moreover, different from the related works, $(k,\delta)$-truss requires the \textit{minimum time span} but not normal duration of triangles to be no greater than $\delta$, which regulates temporal cohesion beyond static cohesion $k$. Because a triangle that ever occurs in a short enough period should be more cohesive in the sense of time than another triangle that has the same duration but no interactions close in time between each pair of its vertices, as illustrated in Fig~\ref{fig:mst}.

Let us consider the following empirical $(k,\delta)$-truss queries.

\begin{example}[Case Study]
    To demonstrate the effectiveness of $(k,\delta)$-truss query, we conduct a case study on Email~\cite{leskovec2005graphanalysis}, a communication network between members of a European research institution. Fig~\ref{fig:case} illustrates four trusses with $k = 16$ and $\delta = \infty$, $200$, $150$, and $100$ respectively. The $(k,\delta)$-truss query can help us to further tailor the static $k$-truss (a) in time dimension. We can see that, the stricter temporal cohesion indeed makes the truss structure more compact. The $k$-truss (a) is composed of several smaller communities. With the decrease of $\delta$, the successive $(k,\delta)$-trusses (b), (c), and (d) become more and more clustered. As an evidence, the clustering coefficient increases from $0.72$ to $0.85$ gradually. In contrast, the duration of whole trusses does almost not change, which means we will not find the more compact trusses like (b), (c), and (d) by using the duration of truss as temporal constraint.
\end{example}

In order to address the $(k,\delta)$-truss query, we propose both index-free and index-based approaches. Unlike the index-free approach that performs a straightforward truss decomposition under the constraint of $\delta$, the index-based approach only needs to scan each edge in the result once, thereby being theoretically time-optimal. Since there could be a great number of $(k,\delta)$-trusses in a temporal graph, we use Temporal Containment Index (TC-Index) or Dual Containment Index (DC-Index) to preserve trusses incrementally. Moreover, for large-scale temporal graphs, we develop two scalable index construction algorithms based on two classical paradigms, truss decomposition~\cite{wang2012trussdecomp,chen2014trussdecomp,che2010trussdecomp,kabir2017paratrussdecomp,kabir2017sharedtrussdecomp} and truss maintenance~\cite{huang2014trusscommunity, wang2019trussmaint, cliu2014trussmaint, luo2020trussmaintenance,suntrussmaintenance}, respectively. Decomposition Based Algorithm (DBA) can compute the incremental edge sets between $(k,\delta)$-truss and $(k,\delta+1)$-truss, and Maintenance Based Algorithm (MBA) can further compute the incremental edge sets between $(k,\delta)$-truss and $(k+1,\delta)$-truss. \textcolor{black}{Lastly, we propose a local search approach to maintain the indexes when new edges or timestamps are inserted into temporal graphs}.

%As a result, DBA can only be used to construct TC-Index, and MBA can be used to construct both TC-Index and DC-Index.

In summary, our contributions are as follows.

\begin{itemize}
    \item Inspired by the latest studies on temporal triangle and motif, we formalize a novel $(k,\delta)$-truss query problem on temporal graphs, with respect to a meaningful temporal triangle metric called minimum time span. The $(k,\delta)$-truss considers both static and temporal cohesion, while the static $k$-truss is its special case for $\delta = \infty$.
    \item We leverage the dual containment relation on $(k,\delta)$-trusses to design compact indexes that store and retrieve all possible $(k,\delta)$-trusses efficiently. Firstly, we present a map-structured index called TC-Index that preserves the incremental edges in $\delta$ dimension. Then, we present a tree-structured index called DC-Index that preserves the globally minimum incremental edges in both $k$ and $\delta$ dimensions. Both indexes provide the optimal query efficiency, and DC-Index is space-optimal when query efficiency cannot be degraded. %based on a logical $(k,\delta)$-truss graph and its minimum-weight arborescence. 
    \item To enable the proposed indexes to scale to large temporal graphs, we follow the line of truss decomposition and truss maintenance respectively to improve the scalability of index construction. DBA decomposes $k$-truss gradually by removing triangles in descending order of minimum time span for each $k$. MBA maintains all edge trussness simultaneously when invalidating triangles in descending order of minimum time span. 
    \textcolor{black}{\item We present an efficient filter-and-verification algorithm for dynamic index maintenance, thereby avoiding to rebuild index from scratch. For an edge update, it gradually narrows the ranges of trussness and temporal threshold, and then identifies a small local subgraph comprised of the edges that may need to be updated in indexes.}
    \item We conduct comprehensive experimental evaluation on eight real-world temporal graphs, on which we observed that triangle counts distribute widely on minimum time span. Our index-based TC-Query and DC-Query process queries in interactive time, and outperform the index-free Online-Query by 2$\sim$4 orders of magnitude. Meanwhile, our indexes can achieve the compression ratio of up to $10^{-4}$, \textcolor{black}{and can be updated significantly faster than rebuilding from scratch}.
\end{itemize}

We present the preliminaries, index-free approach, index-based approaches, index construction, \textcolor{black}{index maintenance}, experimental evaluation, related work, and conclusion in the rest sections respectively.

\section{Preliminaries}\label{sec:prel}

{}Let $\mathcal{G} = (\mathcal{V,E})$ be an undirected static graph, where $\mathcal{V}$ is a set of vertices and $\mathcal{E}\subseteq \mathcal{V}\times \mathcal{V}$ is a set of edges. For each edge $e = (u,v) \in \mathcal{E}$, the pair of vertices $u$ and $v$ may have multiple interactions at different times. We denote the set of timestamps of interaction between $u$ and $v$ by $\tau_{(u,v)}$. Thus, the temporal edge between $u$ and $v$ can be represented by $(u,v,\tau_{(u,v)})$, and the temporal graph can be represented by $\mathcal{G}^t = (\mathcal{V},\mathcal{E},\Gamma)$, where $\Gamma = \{\tau_{(u,v)} | (u,v)\in \mathcal{E}\}$ is the set of nonempty timestamp sets for each edge in $\mathcal{E}$. Fig~\ref{fig:etg} illustrates a temporal graph as our running example. Without loss of generality, we use consecutive natural numbers from 0 until $n$ to denote timestamps in a temporal graph.

Then, let us consider the reasonable definition of $k$-truss for temporal graphs. Given a static graph $\mathcal{G}$, the $k$-truss, denoted by $T_{k}(\mathcal{G})$, is defined as the maximal subgraph of $\mathcal{G}$ in which each edge is contained by at least $k - 2$ triangles, where $k\geq 2$ is a user-specified integer that represents topological cohesion of the subgraph. Obviously, the definition of $k$-truss is closely associated with that of triangle, a more fundamental concept for graph. Inspired by~\cite{kovan2011motif, paran2017motif, liu2019motif, wang2020motif, phasa2021triangle, liu2023motif}, we propose a novel and meaningful definition of temporal triangle in the context of truss study as follows.

\begin{figure}[t!]
    \centering
    \includegraphics[width=0.9\linewidth]{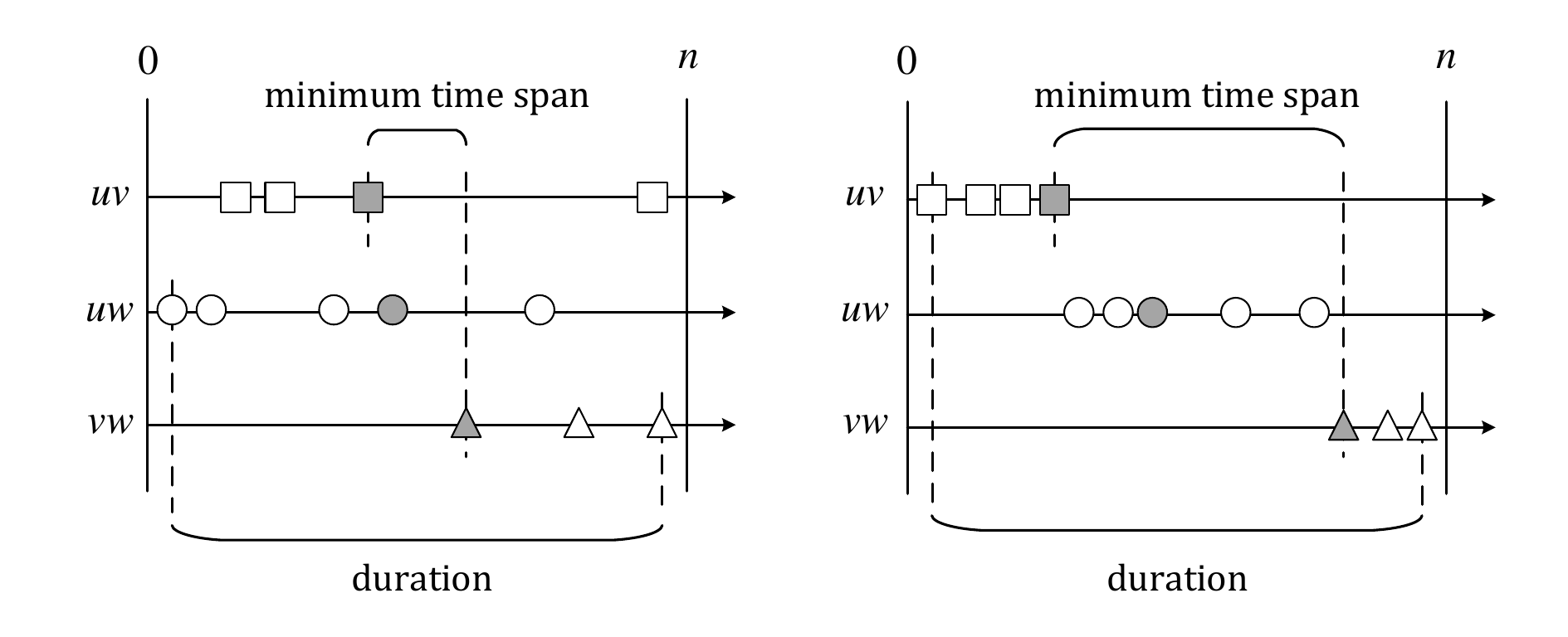}
    \caption{Abstract comparison of different minimum time spans of a triangle $\Delta = \{u,v,w\}$. For each edge like $(u,v)$ of $\Delta$, its timestamps are marked by a kind of symbols like rectangles in a timeline from 0 until $n$. The combination of dark symbols of each kind determines the minimum time span of $\Delta$.}
    \label{fig:mst}
    \vspace{-0.5cm}
\end{figure}

\begin{definition}[Minimum Time Span]
Given a temporal graph $\mathcal{G}^t$, for a triangle $\Delta = \{u,v,w\}$ of $\mathcal{G}$ with $u, v, w \in \mathcal{V}$ and $(u,v), (v,w), (w,u) \in \mathcal{E}$, its minimum time span $\mts(\Delta,\mathcal{G}^t)$, represents the shortest duration of time window in which each two vertices have interaction, namely, $\mts(\Delta,\mathcal{G}^t) = \min\{\max\{|t_1 - t_2|, |t_2-t_3|, |t_3-t_1|\} : t_1 \in \tau_{(u,v)}, t_2 \in \tau_{(v,w)}, t_3 \in \tau_{(w,u)}\}$. When the context is clear, we replace $\mts(\Delta,\mathcal{G}^t)$ by $\mts(\Delta)$. 
\end{definition}

\begin{definition}[$\delta$-triangle]
Given a threshold $\delta$ of minimum time span, a triangle $\Delta$ is called a $\delta$-triangle if $\mts(\Delta) \leq \delta$.
\end{definition}

Intuitively, $\delta$-triangles represent tight bonds in the sense of both topology and time, which require the involved vertices to interact with each other during at least one same short period. As illustrated in Fig~\ref{fig:mst}, the left triangle is considered as a tighter bond than the right one from the perspective of time, though them have the same duration. Because, in the shorter minimum time span, all three vertices may participate in a same event. In contrast, the right triangle represents a typical counterexample. When two of the vertices have contacts, neither of them interacts with the other vertex. Like in a social network, I know both of you but do not know you know each other, which means the three of us are not that close.

On top of $\delta$-triangle, we give the definitions of $\delta$-support of temporal edge and $(k,\delta)$-truss of temporal graph as follows. 

\begin{definition}[$\delta$-support]
    Given a temporal graph $\mathcal{G}^t$ and an integer $\delta \geq 0$, the $\delta$-support of an edge $e = (u,v) \in \mathcal{E}$, denoted by $\delta$-$\support(e)$, is the number of $\delta$-triangles that contain $e$, namely, $|\{\Delta : u,v\in \Delta, \mts(\Delta)\leq \delta \}|$.
\end{definition}

\begin{example}
    Consider an edge $e = (v_2,v_8)$ in Fig~\ref{fig:etg}. The triangles containing $e$ are $\Delta_1 = \{v_2,v_3,v_8\}$ and $\Delta_2 = \{v_2,v_7,v_8\}$. We have $\mts(\Delta_1) = 2$ and $\mts(\Delta_2) = 6$, so that $\delta$-$\support(e)$ is $2$ if $\delta \geq 6$, $1$ if $2 \leq \delta < 6$, or $0$ otherwise.
\end{example}

\begin{definition}[$(k,\delta)$-truss]
    Given a temporal graph $\mathcal{G}^t$, an integer $k \geq 2$, and an integer $\delta \geq 0$, the $(k,\delta)$-truss of $\mathcal{G}^t$, denoted by $T_{k,\delta}$, is defined as the maximal subgraph of $\mathcal{G}^t$ in which the $\delta$-support of each edge $e$ is no less than $k-2$, namely, $\delta$-$\support(e) \geq k-2$ in the temporal subgraph $T_{k,\delta}$. 
\end{definition}

\begin{example}
    In Fig~\ref{fig:etg}, given $k = 4$, we use colored dashed lines to remark the $(k,\delta)$-trusses with different $\delta$. The edges surrounded by black dashed line comprise the largest $(4,6)$-truss. When $\delta$ is decreased to $4$, the edges such as $(v_2,v_8)$ are excluded due to inadequate $\delta$-support, and the edges surrounded by red dashed line comprise the smaller $(4,4)$-truss. Similarly, with the decrease of $\delta$, the $(k,\delta)$-truss shrinks gradually, like the $(4,2)$-truss surrounded by green dashed line and $(4,0)$-truss surrounded by blue dashed line. 
\end{example}

%The $(k,\delta)$-truss is a meaningful cohesive subgraph of temporal graph, which considers both static and temporal cohesion. Note that, the $(k,\infty)$-truss of $\mathcal{G}^t$ is equivalent to the $k$-truss of $\mathcal{G}$, so that $k$-truss can be seen as a special case of $(k,\delta)$-truss in static graphs.

\begin{figure}[t!]
    \centering
    \includegraphics[width=0.7\linewidth]{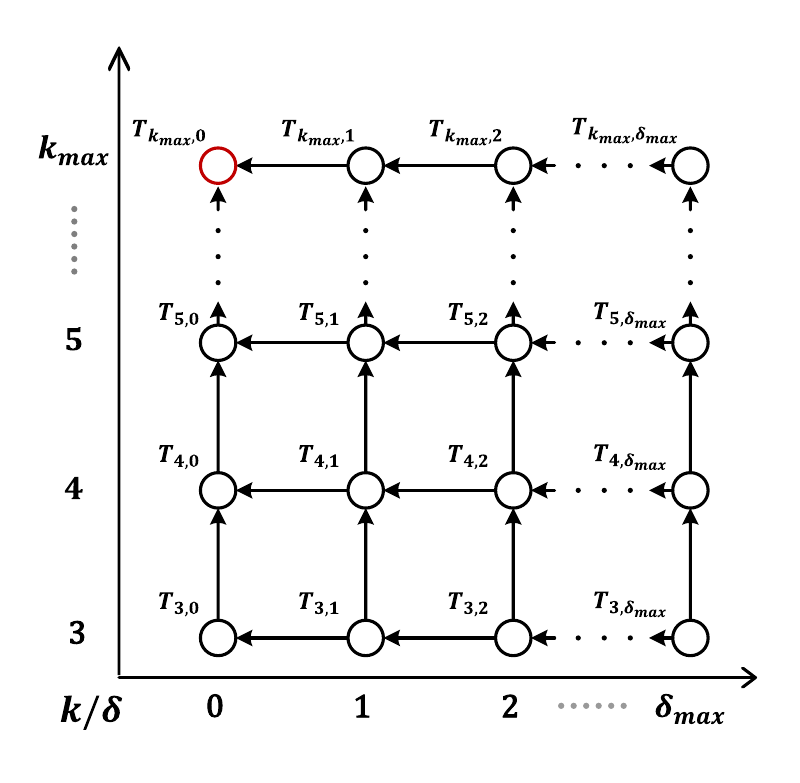}
    \caption{An illustration of $(k,\delta)$-truss graph of a temporal graph, where each arrow $T_{k,\delta} \to T_{k',\delta'}$ denotes that $T_{k',\delta'} \subseteq T_{k,\delta}$.}
    \label{fig:dualcont}
    \vspace{-0.5cm}
\end{figure}

More importantly, similar to the static cohesive subgraphs like $k$-truss and $k$-core, ($k,\delta$)-truss also has the containment property, which is however dual with respect to both $k$ and $\delta$. As illustrated in Fig~\ref{fig:dualcont}, each $(k,\delta)$-truss is contained by both $(k-1,\delta)$-truss and $(k,\delta + 1)$-truss directly. Such containment properties are the keys to developing efficient algorithms for retrieving cohesive subgraphs, such as truss decomposition~\cite{wang2012trussdecomp} and core decomposition~\cite{batage2003coredecomp}. The dual containment property is formally defined as follows.
\begin{property}[Dual Containment]\label{pro:kdtruss}
    For any two $(k,\delta)$-trusses $T_{k, \delta}$ and $T_{k',\delta'}$ of a temporal graph $\mathcal{G}^t$, $T_{k, \delta}$ is a subgraph of $T_{k',\delta'}$, denoted by $T_{k, \delta} \subseteq T_{k', \delta'}$, if $k' \leq k$ and $\delta' \geq \delta$.
\end{property}

In this paper, we aim to address the following problem.

\begin{problem}
    Given a temporal graph $\mathcal{G}^t$, an integer $k \geq 2$, and an integer $\delta \geq 0$, find the $(k,\delta)$-truss $T_{k,\delta}$ of $\mathcal{G}^t$.
\end{problem}

\section{Index-Free Approach}\label{sec:indexfree}

In this section, we first propose a straightforward online solution derived from the classic truss decomposition~\cite{wang2012trussdecomp} as a baseline. The pseudo code of online solution can be found in the conference version~\cite{Hu2024kdtruss}. It first computes the $\delta$-support of each edge in $\mathcal{E}$, and pushes all edges into a priority queue $Q$ in ascending order of $\delta$-support. Then, it performs an edge peeling process that iteratively removes the edge with the minimum $\delta$-support from queue until the $\delta$-supports of all rest edges are no less than $k - 2$. Upon the removal of an edge $e$, for the other edges in each same triangle $\Delta$ with $e$, we will keep their $\delta$-support unchanged if $\mts(\Delta)$ is greater than $\delta$ because they are never counted in the first place, or decrease their $\delta$-support otherwise. After the peeling process, the remaining edges in the queue comprise the target $(k,\delta)$-truss.

%\begin{algorithm}[t!]
%    \caption{Online-Query}\label{alg:online}
%    Please see the conference version~\cite{Hu2024kdtruss}.
%\end{algorithm}

\noindent{\textbf{Correctness}}. The correctness is obvious as long as the correctness of truss decomposition holds. %iteratively remove an edge with lowest $\delta$-support in $\mathcal{G}$ and correctly update the $\delta$-support of other edges sharing a triangle with $e$ until the $\delta$-support of all edges in $\mathcal{G}$ is not less than $k$ - 2. After the algorithm ends, the $\delta$-support of edges in the subgraph $T_{k,\delta}$ return by algorithm is greater than or equal to $k$ - 2. Thus, Algorithm1 is correct. 

\noindent{\textbf{Complexity.}} The original truss decomposition algorithm takes $O(\sum_{(u,v) \in \mathcal{E}} \min\{\degree(u),\degree(v)\})$ time, where $\degree(u)$ is the degree of $u$ in $\mathcal{G}$. Compared with that, the main extra time cost of our algorithm is to compute $\mts(\Delta)$ for each triangle in $\mathcal{G}^t$. Let $\overline{|\tau|}$ denote the average number of timestamps associated with an edge $(u,v)$ and $|\Delta|$ denote the total number of triangles. The time cost of computing $\mts(\Delta)$ for a single triangle is $O(\overline{|\tau|})$ if $\tau_{(u,v)}$ is ordered. Thus, the time complexity of our algorithm is $O(\sum_{(u,v) \in \mathcal{E}} \min\{\degree(u),\degree(v)\} + \overline{|\tau|} \cdot |\Delta|)$.

\section{Index-Based Approach}\label{sec:indexbased}

The complexity of the above index-free algorithm is at least sub-quadratic to the number of edges for a temporal graph, and thus is infeasible for real-time processing when the graph is large. Therefore, we propose index-based approaches to efficiently answer $(k,\delta)$-truss queries in this section.

\subsection{TC-Index}

\subsubsection{Index Structure}

A basic idea of indexing is to preserve all possible $(k,\delta)$-trusses for a temporal graph, which can answer any query in the optimal time due to precomputation. However, directly preserving all possible $(k,\delta)$-trusses takes $O(k_{max} \cdot \delta_{max} \cdot |\mathcal{E}|)$ space in the worst case, where $k_{max}$ and $\delta_{max}$ are the maximum values of $k$ and $\delta$ respectively for a given temporal graph. It means the index could be thousands of times or even larger than the graph itself in practice. Consequently, we propose a Temporal Containment Index (TC-Index) that adopts an incremental storage scheme to reduce the index size. With only a little compromise of query efficiency compared with the uncompressed index, the size of TC-Index is reduced to $O(k_{max} \cdot (|\mathcal{E}| + \delta_{max}))$. 

Before introducing TC-Index, let us consider the following pilot concept firstly.

\begin{definition}[$k$-Span]\label{def:kspan}
    Given a temporal graph $\mathcal{G}^t$ and a support threshold $k$, the $k$-span of an edge $e \in \mathcal{E}$ is an integer $\delta = k\text{-}\kspan(e, \mathcal{G}^t)$, such that (i) the $(k,\delta)$-truss contains $e$ and (ii) the $(k,\delta')$-truss does not contain $e$ for any $\delta' < \delta$. When the content is clear, we replace $k\text{-}\kspan(e,\mathcal{G}^t)$ by $k\text{-}\kspan(e)$.
\end{definition}

\begin{property}\label{pro:kspan}
    The $k$-span of edges in the $(k,\delta)$-truss is no greater than $\delta$.
\end{property}

With Property~\ref{pro:kdtruss} and~\ref{pro:kspan} , we only need to consider edges in the $k$-truss of $\mathcal{G}$ whose $k$-span is no greater than $\delta$ in order to find the $(k,\delta)$-truss of $\mathcal{G}^t$, because the $(k,\delta)$-truss of $\mathcal{G}^t$ is certainly a subgraph of the $k$-truss of $\mathcal{G}$.

Based on the above observation, for a temporal graph $\mathcal{G}^t$, TC-Index maintains a map structure $\mathcal{I}_k$ = $(\mathcal{E}_k, \mathcal{D}_k)$ for each possible $k$ value. $\mathcal{E}_k$ is a sequence that preserves edges of the $k$-truss in descending order of $k$-span, and $\mathcal{D}_k$ is an index that records the unique $k$-spans and the offsets of the first edges with the corresponding $k$-span in $\mathcal{E}_k$. Then, TC-Index $\mathcal{I} = (\mathcal{I}_3, \mathcal{I}_4, \cdots, \mathcal{I}_{k_{max}})$ is comprised of the map structures for all possible $k$. Note that TC-Index does not store the map structures for $k \leq 2$ because the $(2,\delta)$-truss is actually the entire temporal graph, regardless of $\delta$. 

\begin{figure*}
    \centering
    \includegraphics[width=0.9\linewidth]{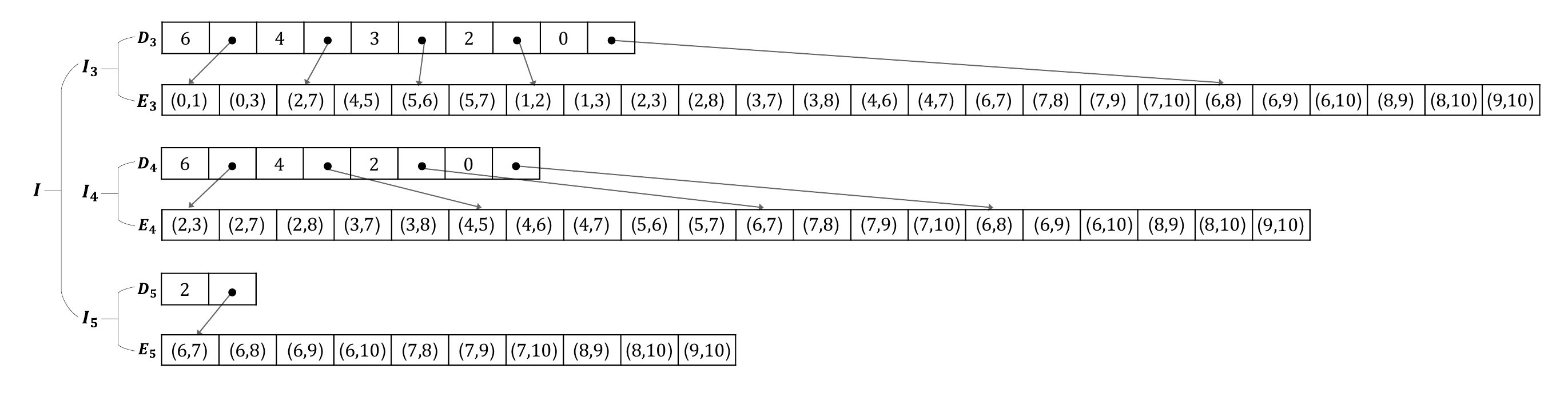}
    \caption{An example of TC-Index, in which $(v_i,v_j)$ is represented by $(i,j)$.}
    \label{fig:tcindex}
    \vspace{-0.3cm}
\end{figure*}

\begin{example}
    The TC-Index for the temporal graph in Fig~\ref{fig:etg} is given in Fig~\ref{fig:tcindex}. Three are three map structures $\mathcal{I}_3$, $\mathcal{I}_4$, and $\mathcal{I}_5$. Within $\mathcal{I}_4$, there are four unique $k$-spans $6$, $4$, $2$, and $0$ in $\mathcal{D}_4$, which have pointers to the first edges with these $k$-spans in $\mathcal{E}_4$. We can see the edges $(v_2,v_3)$, $(v_2,v_7)$, $(v_2,v_8)$, $(v_3,v_7)$, and $(v_3,v_8)$ have the $k$-span $6$ when $k = 4$.
\end{example}

\begin{theorem}\label{the:tcspace}
    For a temporal graph $\mathcal{G}^t$, the size of TC-Index is bounded by $O(k_{max} \cdot (|\mathcal{E}| + \delta_{max}))$.
\end{theorem}

\subsubsection{Query Processing}
The TC-Index based query algorithm is named TC-Query, whose pseudo code can be found in the conference version~\cite{Hu2024kdtruss}. For the given $k$ and $\delta$, TC-Query first retrieves $\mathcal{I}_k$ = $(\mathcal{E}_k, \mathcal{D}_k)$ from TC-Index. Then, it finds the maximum $\delta'$ in $\mathcal{D}_k$ with $\delta' \leq \delta$, and locates the position in $\mathcal{E}_k$ with respect to the offset associated with $\delta'$ in $\mathcal{D}_k$. Lastly, it scans $\mathcal{E}_k$ from the position until the end, and all scanned edges comprise $T_{k,\delta}$.

\begin{example}
    Consider the $(k,\delta)$-truss query with $k = 4$ and $\delta = 1$. In $\mathcal{D}_4$, the first $k$-span no greater than $\delta$ is $0$. Thus, we locate the edge $(v_6,v_8)$ in $\mathcal{E}_4$ and start to scan the following edges. Lastly, we get the edges of $(4,1)$-truss, namely, $(v_6,v_8)$, $(v_6,v_9)$, $(v_6,v_{10})$, $(v_8,v_9)$, $(v_8,v_{10})$, and $(v_9,v_{10})$, which can be verified in Fig~\ref{fig:etg} (see the blue edges).
\end{example}

\begin{theorem}\label{the:tctime}
    TC-Query computes the edge set of $T_{k,\delta}$ in at most $O(\log\delta_{max} + |T_{k,\delta}|)$ time, where $|T_{k,\delta}|$ denotes the number of edges in $T_{k,\delta}$.
\end{theorem}
Compared with the optimal time of processing $(k,\delta)$-truss query that is certainly $O(|T_{k,\delta}|)$, TC-Query is almost optimal since $\log\delta_{max}$ is usually much less than $|T_{k,\delta}|$.

%\begin{algorithm}[t!]
%    \caption{TC-Query}\label{alg:tcquery}
%    Please see the conference version~\cite{Hu2024kdtruss}.
%\end{algorithm}

\begin{figure*}
    \centering
    \includegraphics[width=0.9\linewidth]{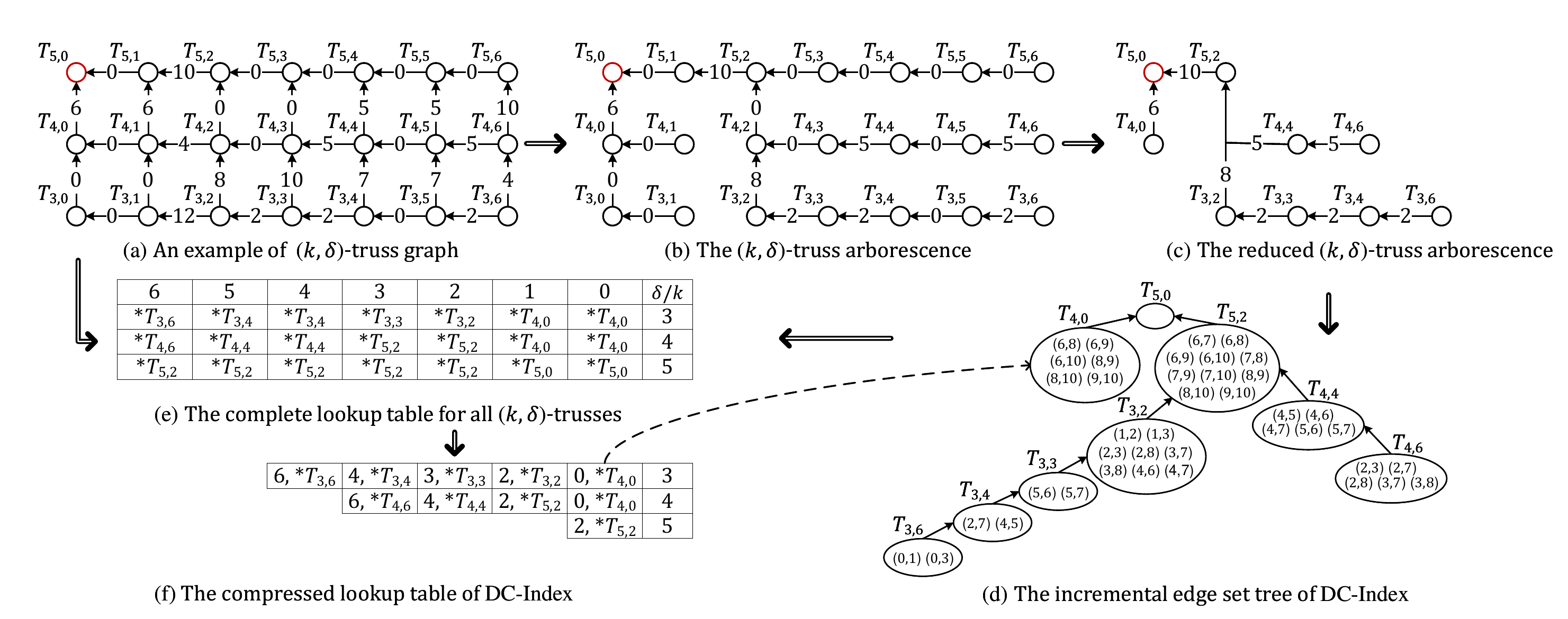}
    \caption{An example of DC-Index, in which $(v_i,v_j)$ is represented by $(i,j)$.}
    \label{fig:dcindex}
    \vspace{-0.2cm}
\end{figure*}

\subsection{DC-Index}

\subsubsection{Index Structure}

Although TC-Index exploits temporal containment to achieve incremental storage, it is still not space-efficient enough as it only takes into account one aspect of Property~\ref{pro:kdtruss}. To fully exploit the property, we propose an advanced tree-structured index called Dual Containment Index (DC-Index). DC-Index is space-optimal while guaranteeing the same order of query efficiency as TC-Index. Specifically, DC-Index is derived in the following steps.

\begin{definition}[$(k,\delta)$-Truss Graph]
    Given a temporal graph $\mathcal{G}^t$, we define a $(k,\delta)$-truss graph as a directed weighted graph $G = (V, E, w)$, where $V = \{T_{k,\delta} : 3\leq k\leq k_{max}, 0\leq \delta \leq \delta_{max}\}$ is the set of all possible $(k,\delta)$-trusses of $\mathcal{G}^t$, $E = E_v \cup E_h = \{(T_{k,\delta}, T_{k+1,\delta}) : 3\leq k\leq k_{max}-1, 0\leq \delta \leq \delta_{max}\}\cup \{(T_{k,\delta}, T_{k,\delta-1}) : 3\leq k\leq k_{max}, 1\leq \delta \leq \delta_{max}\}$ is a subset of dual containment relation on $V$, and $w : E \mapsto \mathbb{N}$ indicates the number of incremental edges between two connected trusses.
\end{definition}

Fig~\ref{fig:dcindex}(a) illustrates the $(k,\delta)$-truss graph of our example temporal graph. Intuitively, we call the edges in $E_v$ as vertical edges and the edges in $E_h$ as horizontal edges. For each edge in this graph, if the sink truss has been stored, the cost of incrementally preserving the source truss is the weight of edge. 

\begin{definition}[$(k,\delta)$-Truss Arborescence]\label{def:ta}
    Given a $(k,\delta)$-truss graph $G$, we derive a $(k,\delta)$-truss arborescence $A$ from it by removing the outgoing edge $(T_{k,\delta}, T_{k+1,\delta})$ or $(T_{k,\delta}, T_{k,\delta-1})$ with greater weight for each truss $T_{k,\delta}$. Note that, if $T_{k,\delta}$ has one or none outgoing edge, no edge will be removed.
\end{definition}

Fig~\ref{fig:dcindex}(b) illustrates the $(k,\delta)$-truss arborescence, which is actually a minimum-weight directed spanning tree. For each truss, there is a single directed path that leads to the root, namely, $(k_{max},0)$-truss. %Meanwhile, the $(k,\delta)$-truss arborescence represents the compactest incremental storage scheme with no extra decompression cost.

\begin{definition}[Reduced $(k,\delta)$-Truss Arborescence]\label{def:rta}
    Given a $(k,\delta)$-truss arborescence $A$, we reduce it to another arborescence $A^-$ by (i) removing each truss and its outgoing edge if the edge weight is zero and (ii) reconnecting each truss to the next truss on its original path in $A$ if its sink truss is removed.
\end{definition}

Fig~\ref{fig:dcindex}(c) illustrates the reduced $(k,\delta)$-truss arborescence. Obviously, each omitted truss can still be retrieved because there is certainly another truss identical to it remained. 
Then, we use an incremental storage scheme to preserve all possible $(k,\delta)$-trusses according to the reduced $(k,\delta)$-truss arborescence. As illustrated in Fig~\ref{fig:dcindex}(d), the \textbf{incremental edge set tree} of DC-Index is logically equivalent to the reduced $(k,\delta)$-truss arborescence, and preserves the Incremental Edge Sets (IESes) between trusses in each node of the tree. Compared with TC-Index, DC-Index surely has fewer redundant edges. For example, the total number of edges in Fig~\ref{fig:dcindex}(d) is 40, and in contrast, the total number of edges in Fig~\ref{fig:tcindex} is 54. Actually, the space cost of DC-Index is the minimum under a precondition.

% \begin{figure}[t!]
%     \centering
%     \includegraphics[width=0.7\linewidth]{DC.pdf}
%     \caption{An example of DC-Index, in which $(v_i,v_j)$ is represented by $(i,j)$.}
%     \label{fig:dcindex}
% \end{figure}

\begin{theorem}\label{the:dcspace}
    Given a temporal graph $\mathcal{G}^t$, the incremental edge set tree of DC-Index is space-optimal for preserving all possible $(k,\delta)$-trusses of $\mathcal{G}^t$, when the efficiency of retrieving a specific $(k,\delta)$-truss cannot be degraded.
\end{theorem}
Moreover, to retrieve a specific $(k,\delta)$-truss in the tree, DC-Index uses a lookup table to record the pointers to tree nodes. As shown in Fig~\ref{fig:dcindex}(e), the cell in $k$ row and $\delta$ column contains the pointer to the tree node that represents an identical truss of $(k,\delta)$-truss. For example, to lookup $T_{3,0}$, the pointer to $T_{4,0}$ in the tree is returned. We further compress the lookup table by skipping the consecutively repeating pointers for each row. As illustrated in Fig~\ref{fig:dcindex}(f), when $k=5$, there are only two unique pointers to $T_{5,0}$ and $T_{5,2}$ respectively, and we only record these two pointers with their smallest column ids (namely, 0 and 2) in the row. For $0 < \delta < 2$ or $2 < \delta \leq 6$, it is easy to know the corresponding $T_{5,\delta}$ is identical to $T_{5,0}$ or $T_{5,2}$.

% Moreover, to retrieve a specific $(k,\delta)$-truss in the tree, DC-Index uses a lookup table to record the pointers to tree nodes. As illustrated in Fig~\ref{fig:dcindex}(e), the cell in $k$ row and $\delta$ column contains the pointer to the tree node that represents an identical truss of $(k,\delta)$-truss. For example, to lookup $T_{3,0}$, the pointer to $T_{4,0}$ in the tree is returned. We further compress the lookup table by skipping the consecutively repeating pointers for each row. As illustrated in Fig~\ref{fig:dcindex}(f), when $k=5$, there are only two unique pointers to $T_{5,0}$ and $T_{5,2}$ respectively, and we only record them and their smallest column ids (namely, 0 and 2) in the row. For $0 < \delta < 2$ or $2 < \delta \leq 6$, it is easy to know the corresponding $T_{5,\delta}$ is identical to $T_{5,0}$ or $T_{5,2}$.

\subsubsection{Query Processing}

% The DC-Index based query algorithm is named DC-Query, which is similar to TC-Query. For the given $k$ and $\delta$, it firstly finds the node with maximum $\delta' \leq \delta$ in row $k$. \textcolor{black}{Then, we traverse the path from this node to the root via the \textit{Father} pointer, and the union of all traversed edge sets is the edge set of $T_{k,\delta}$. The pseudo code is omitted.}
The DC-Index based query algorithm is named DC-Query, which is similar to TC-Query. For the given $k$ and $\delta$, it firstly finds the maximum $\delta' \leq \delta$ in row $k$ and gets the pointer to a node of incremental edge set tree. Then, we traverse the path from this node to the root, and the union of all traversed edge sets is the edge set of $T_{k,\delta}$. The pseudo code is omitted.

\begin{theorem}
    DC-Query is as efficient as TC-Query.
\end{theorem}

\section{Index construction}\label{sec:ic}
In this section, we address the scalable construction of TC/DC-Index on large-scale temporal graphs. For that, we propose two algorithms based on two classic paradigms, namely, truss decomposition~\cite{wang2012trussdecomp,chen2014trussdecomp,che2010trussdecomp,kabir2017paratrussdecomp,kabir2017sharedtrussdecomp} and truss maintenance~\cite{huang2014trusscommunity, wang2019trussmaint, cliu2014trussmaint, luo2020trussmaintenance,suntrussmaintenance}, respectively.

\subsection{TC-Index Construction based on Decomposition}

Intrinsically, the construction of TC-Index can be addressed by computing the Incremental Edge Sets (IESes) between each pair of $(k,\delta)$-truss and $(k,\delta+1)$-truss for a temporal graph. Since these IESes correspond to the horizontal edges ($E_h$) of the $(k,\delta)$-truss graph illustrated in Fig~\ref{fig:dualcont}, we call them Horizontal IES (H-IES), The index construction algorithm needs to reduce the redundant part in the computation of different H-IESes to scale to large temporal graphs. A basic strategy is to exploit the dual containment property to decrementally decompose $(k,\delta)$-trusses in a particular order with respect to $k$ and $\delta$, like the existing temporal $k$-core decomposition~\cite{Yang2023temporalcore}.

However, temporal $k$-core decomposition does not involve the extra metric like the minimum time span of triangle, which may be needed repeatedly for evaluating $\delta$-supports of edges. Thus, to eliminate the repeating computation of minimum time span, we design the following data structure to store $\mts(\Delta)$ for each $\Delta$ in a temporal graph.

\begin{definition}[$\delta$-triangle List]
    Given a temporal graph $\mathcal{G}^t$, a $\delta$-triangle list is a list of triangle sets $(S^{\Delta}_0, S^{\Delta}_1, \cdots, S^{\Delta}_{\delta_{max}})$, where $S^{\Delta}_\delta$ is the set of all triangles in $\mathcal{G}^t$ whose minimum time spans are exactly $\delta$ with $0\leq \delta \leq \delta_{max}$. 
\end{definition}

%\begin{algorithm}[t!]
%    \SetKwFunction{dch}{decomph}
%    \SetKwFunction{dcv}{decompv}
%    \caption{Decomposition Based Algorithm (DBA)}\label{alg:dba}
%    Please see the conference version~\cite{Hu2024kdtruss}.
%\end{algorithm}

Then, for each $k$, we decrementally induce each $(k,\delta)$-truss from $(k,\delta+1)$-truss, so that the H-IES between them can be obtained. The pseudo code of our Decomposition Based Algorithm (DBA) can be found in the conference version~\cite{Hu2024kdtruss}. Initially, we enumerate all triangles and evaluate their minimum time span, for building the $\delta$-triangle list. Let $X^{\Delta}_k$ denote the set of triangles of $k$-truss. Obviously, the 2-truss is $\mathcal{G}^t$ itself, and thereby $X^{\Delta}_2$ is the union of all sets in $k$-triangle list. Then, for each $k$ with $3\leq k \leq k_{max}$, we decompose $T_{k,\delta_{max}} = T_k$ gradually until $T_{k,0}$ is obtained, and collect the H-IESes. Specifically, each iteration starts with obtaining $T_k$ of $\mathcal{G}$ by a traditional truss decomposition function \dcv{}, since the $(k,\delta)$-truss is always a subgraph of $k$-truss. In particular, we can also obtain the set of deleted triangles $X^{\Delta}$ by \dcv{}. Thus, the triangle set of $k$-truss $X^{\Delta}_{k}$ can be computed by removing $X^{\Delta}$ from $X^{\Delta}_{k-1}$. Then, for each $\delta$ from $\delta_{max}-1$ until 0, the $(k,\delta)$-truss is induced by a new decomposition function \dch{}, which invalidates the triangles in $X^{\Delta}_k$ from the $(k,\delta+1)$-truss. In this function, the H-IES can also be obtained. The details of functions are also included in the conference version~\cite{Hu2024kdtruss}.

%\noindent{\dcv{}}. This function could be any existing truss decomposition algorithm like~\cite{wang2012trussdecomp}. So the details are omitted.

%\noindent{\dch{}}. This function mainly consists of two phases. In the first phase (Lines 13-18), for each triangle whose minimum time span is $\delta$ in $X^{\Delta}_k$, we decrease the $\delta$-support of its edges by 1 and push the edges whose new $\delta$-support is less than $k-2$ into a candidate container $Q$, which collects edges waiting to be deleted. In the second phase (Lines 19-28), for each edge $e$ in $Q$, we remove it from $T_{k,\delta}$, and meanwhile add it to the H-IES $R$. Upon the removal of $e$, if there is any other triangle $\Delta \in X^{\Delta}_k$ that contains $e$, we process the other edges of $\Delta$ than $e$ as in the first phase.

\begin{example}
Consider the temporal graph in Fig~\ref{fig:etg}. Since the greatest minimum time span of its triangles $\delta_{max} = 6$, the $4$-truss is actually the $(4,6)$-truss surrounded by black dashed line. For $k = 4$, we can get the $4$-truss from previous $3$-truss by a traditional truss decomposition, and start to decompose $(4,6)$-truss from $\delta = 5$. The triangle $\Delta = \{v_2,v_7,v_8\} \in X^{\Delta}_4$ with $\mts(\Delta) = 6$ will be invalidated because we are trying to obtain $(4,5)$-truss currently. Thus, the $\delta$-supports of its edges are decreased by $1$. Then, we have both $5$-$\support(v_2,v_7)$ and $5$-$\support(v_2,v_8) = 1 < 4-2$, so $(v_2,v_7)$ and $(v_2,v_8)$ are pushed into $Q$ and waiting to be deleted. On the removal of edge $(v_2,v_7)$, the triangle $\{v_2, v_3, v_7\}$ containing it will be broken, so that the other two edges $(v_2,v_3)$ and $(v_3,v_7)$ also need to be checked like $(v_2,v_7)$. Iteratively, the edges not belong to $(4,5)$-truss are all deleted and pushed into the H-IES between $(4,6)$-truss and $(4,5)$-truss, until $Q$ is empty. Similarly, when $\delta$ decreases gradually, the edges remarked by red color and green color will be deleted respectively, until only the blue edges that comprise the $(4,0)$-truss remain.
\end{example}

\noindent{\textbf{Correctness.}} We only discuss the correctness of \dch{} here, as the other parts of DBA are straightforward. Due to Property~\ref{pro:kdtruss}, we can induce the $(k,\delta)$-truss from the $(k,\delta+1)$-truss. It is easy to know, we can finish that by invalidating all triangles in $X^{\Delta}_k$ whose minimum time span is greater than $\delta$. To avoid redundant computation, we have a trick that is to only invalidate the triangle whose minimum time span is exactly $\delta + 1$, because the other triangles have been invalidated during the previous calls of \dch{}. The rest decomposition procedure is correct obviously.

\noindent{\textbf{Complexity.}} The time complexity of building the $\delta$-triangle list is $O(\sum_{(u,v) \in \mathcal{E}} \min\{\deg(u),\deg(v)\} + \overline{|\tau|} \cdot |\Delta|)$, which is the same as the online algorithm. The total time cost of calling \dcv{} is $O(\sum_{(u,v) \in \mathcal{E}} \min\{\deg(u),\deg(v)\}$. For the inner loop with a specific $k$, the time complexity is $O(|X_k^{\Delta}| + \sum_{(u,v) \in T_k} \min\{\deg(u),\deg(v)\})$, since each triangle in $T_k$ is invalidated at most once and each edge in $T_k$ is visited at most once. Thus, the total time cost of calling \dch{} is $O(\sum_{k=3}^{k_{max}} (|X_k^{\Delta}| + \sum_{(u,v) \in T_k} \min\{\deg(u),\deg(v)\}))$, which dominates the total time cost of DBA. 

\subsection{TC/DC-Index Construction based on Maintenance}

DBA can only produce H-IES but not IES between each pair of $(k,\delta)$-truss and $(k+1,\delta)$-truss, which are called Vertical IES (V-IES) for corresponding to the vertical edges ($E_v$) of the $(k,\delta)$-truss graph in Fig~\ref{fig:dcindex}, and thereby is not efficient for constructing DC-Index. Thus, we propose another Maintenance Based Algorithm (MBA). MBA can construct both TC-Index and DC-Index, and is more efficient than DBA.

The traditional truss maintenance problem is to update the trussness of edges when edges are inserted or deleted, which has been widely studied~\cite{huang2014trusscommunity, wang2019trussmaint, cliu2014trussmaint, luo2020trussmaintenance}. Different from that, we mainly focus on truss maintenance when triangles are validated or invalidated with respect to minimum time span. Moreover, similar to edge-oriented maintenance, updating trussness for triangle invalidation is much more efficient than triangle validation. Therefore, we only maintain edge trussness when a triangle becomes invalid due to the decease of $\delta$, with respect to the following observations.

\begin{lemma}\label{lem:trussness1}
    Given a graph $\mathcal{G}$, the trussness of any edge in $\mathcal{E}$ can be decreased by at most $1$ if a triangle $\Delta$ of $\mathcal{G}$ gets invalid.
\end{lemma}

With Lemma~\ref{lem:trussness1}, we only need to identify the edges whose trussness will be affected by triangle invalidation and decrease their trussness by $1$ for maintenance.

\begin{definition}[$k$-triangle]\label{def:ktriangle}
    Given a graph $\mathcal{G}$, a triangle $\Delta$ of $\mathcal{G}$ is a $k$-triangle if the minimum edge trussness in the triangle is $k$. We say the level of $\Delta$ is $k$, denoted by $L(\Delta, \mathcal{G})$. When the content is clear, we replace $L(\Delta,\mathcal{G})$ by $L(\Delta)$.
\end{definition}
  
\begin{lemma}\label{lem:trussness2}
    For any $k$-triangle $\Delta$ of $\mathcal{G}$, the trussness of an edge $e \in \mathcal{E}$ will not be updated when $\Delta$ gets invalid if $\trn(e) \neq k$.
\end{lemma}

\begin{lemma}\label{lem:trussness3}
    For any $k$-triangle $\Delta$ of $\mathcal{G}$, the trussness of an edge $e \,{\in}\, \mathcal{E}$ may be updated when $\Delta$ gets invalid, if $\trn(e) = k$ and one of following conditions is satisfied: (i) $e \in \Delta$ or (ii) $e \notin \Delta$ and $\exists e' \in \Delta$ such that $\trn(e')$ = $k$ and $e$ is connected with $e'$ though a series of $k$-triangles sharing common edges.    
\end{lemma}

%\begin{algorithm}[t!]
%    \caption{Triangle Invalidation}\label{alg:single}
%    Please see the conference version~\cite{Hu2024kdtruss}.
%\end{algorithm}

With Lemma~\ref{lem:trussness2} and~\ref{lem:trussness3}, we develop an algorithm to maintain the edge trussness for a single triangle invalidation, which draws inspiration from the removal algorithm proposed by~\cite{wang2019trussmaint}. The most important trick is that, for each edge $e \in \mathcal{E}$, we maintains a stricter $k$-support $ks(e)$ = $|\{\Delta \,{:}\, e\in \Delta, L(\Delta) = \trn(e)\}|$, which is the number of $\trn(e)$-triangles containing $e$. In the $\trn(e)$-truss, $ks(e)$ actually becomes the number of triangles that contain $e$, so that $ks(e)$ is no less than $\trn(e) - 2$. 

The pseudo code of triangle invalidation algorithm can be found in the conference version~\cite{Hu2024kdtruss}. After the invalidation of input $k$-triangle $\Delta$, we firstly decrease the $k$-support of each edge $e$ of $\Delta$ with $\trn(e)$ = $k$ according to the case (i) of Lemma~\ref{lem:trussness3}, and pushes $e$ to a queue $Q$ if $ks(e)$ is no longer greater than $\trn(e)-1$, implying that its trussness will decrease. Then, for each edge $e \in Q$, since the decrease of $\trn(e)$ can further result in the decrease of trussness of other edges in a same triangle with $e$, we perform a breadth-first search to find such edges according to the case (ii) of Lemma~\ref{lem:trussness3}. Lastly, the algorithm returns the set of edges whose trussness is decreased by triangle invalidation.

Then, let us consider how to construct DC-Index through edge trussness maintenance. Since the construction of DC-Index requires to compute both V-IES and H-IES, we present the methods respectively.

\noindent{\textbf{V-IES}}. The V-IES between $(k,\delta)$-truss and $(k+1,\delta)$-truss is simply the set of edges whose trussness is $k+1$ when the triangles whose minimum time span is greater than $\delta$ have been invalidated. Thus, we can obtain all V-IESes by invalidating triangles in descending order of minimum time span gradually, and maintain the edge trussness simultaneously. 

\noindent{\textbf{H-IES}}. The H-IES between $(k,\delta)$-truss and $(k,\delta - 1)$-truss can be further obtained with respect to the following observation.

\begin{lemma}\label{lem:hies}
    When all triangles whose minimum time span is greater than $\delta$ have been invalidated, for an edge $e$ with $\trn(e) = k$, it belongs to the H-IES between $(k,\delta)$-truss and $(k,\delta - 1)$-truss if $\trn(e)$ decreases after any triangle $\Delta$ with $\mts(\Delta) = \delta$ has been further invalidated.
\end{lemma}

With Lemma~\ref{lem:hies}, we can obtain both V-IES and H-IES in the procedure of truss maintenance. Specifically, after invalidating each triangle whose minimum time span is $\delta$, we only need to check the edges whose trussness is exactly the level of this triangle according to Lemma~\ref{lem:trussness2}, and add edges whose trussness has decreased into the corresponding H-IES.

\begin{figure}[t!]
    \centering
    \includegraphics[width=1.0\linewidth]{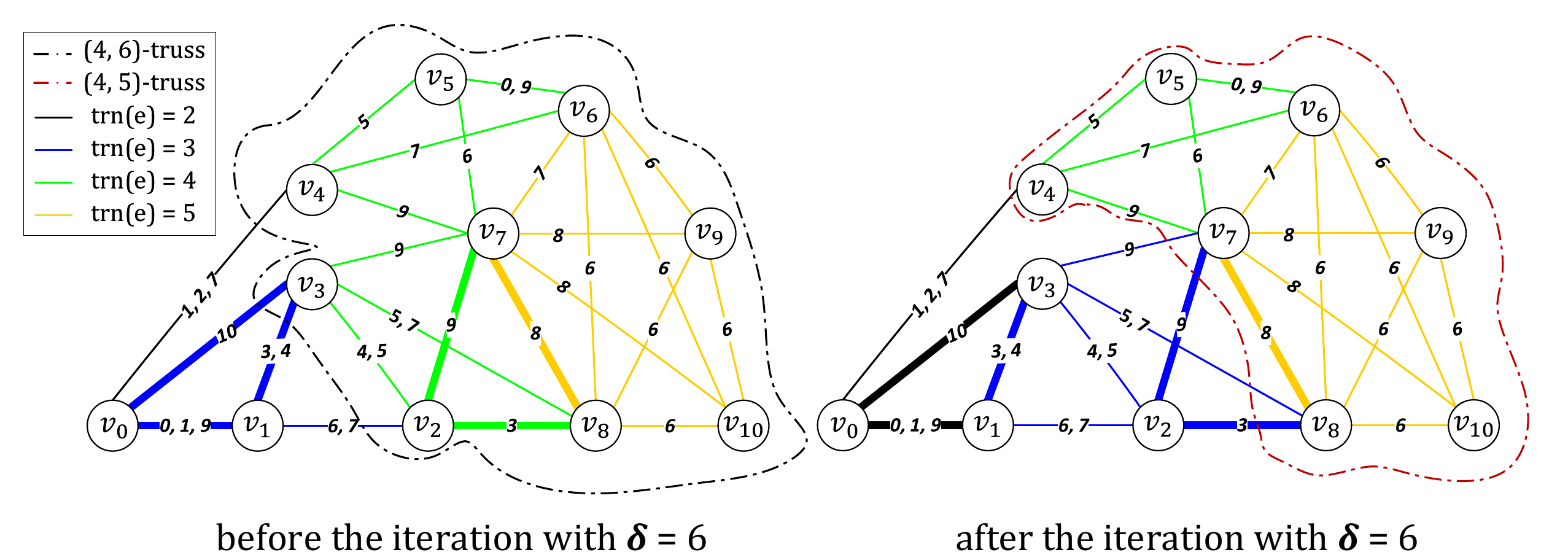}
    \caption{An example of executing MBA before and after $\delta = 6$.}
    \label{fig:mbag}
    \vspace{-0.3cm}
\end{figure}

%\begin{algorithm}[t!]
%    \caption{Maintenance Based Algorithm (MBA)}\label{alg:mba}
%    Please see the conference version~\cite{Hu2024kdtruss}.
%\end{algorithm}

The pseudo code of MBA can be found in the conference version~\cite{Hu2024kdtruss}. We first build the $\delta$-triangle list and compute the initial edge trussness of $\mathcal{G}$. Then, we enumerate $\delta$ gradually in descending order and compute the H-IES and V-IES respectively. For a specific $\delta$, obtaining the V-IES is straightforward since the edge trussness is already known. Moreover, we invalidate each triangle $\Delta$ with $\mts(\Delta) = \delta$ by calling the triangle invalidation algorithm. For each edge returned, we add $e$ to the H-IES between $(k,\delta)$-truss and $(k,\delta-1)$-truss, where $k$ is the level of $\Delta$. 

\begin{example}
    Consider running MBA on the temporal graph in Fig~\ref{fig:etg}. As illustrated in Fig~\ref{fig:mbag}, we remark the trussness of edges by different colors. The initial trussness before iterations is shown on the left. We highlight the $(4,6)$-truss comprised of green and yellow edges because $\delta_{max} = 6$. Then, in the iteration with $\delta = 6$, the triangle $\Delta = \{v_2, v_7, v_8\}$ is invalidated because $\mts(\Delta) = 6$. According to Lemma~\ref{lem:trussness2}, $(v_7,v_8)$ will be not affected since $L(\Delta) = 4 \neq \trn(v_7,v_8) = 5$. According to the case (i) of Lemma~\ref{lem:trussness3}, $(v_2,v_7)$ and $(v_2,v_8)$ will be affected. According to the case (ii) of Lemma~\ref{lem:trussness3}, $(v_2,v_3)$, $(v_3,v_7)$, and $(v_3,v_8)$ will be affected. Lastly, according to Lemma~\ref{lem:trussness1}, the trussness of affected edges is decreased by 1. The updated trussness is shown on the right, and the current green and yellow edges comprise the highlighted $(4,5)$-truss.
\end{example}

\noindent{\textbf{Correctness.}} The correctness of MBA is established on the corretness of Lemma~\ref{lem:trussness1},~\ref{lem:trussness2},~\ref{lem:trussness3}, and~\ref{lem:hies}.

\noindent{\textbf{Complexity.}} For brevity, we only compare the difference of dominant time cost between DBA and MBA here. For the quantity of visited edges (pushed into $Q$), two algorithms are equivalent, though DBA enumerates $k$ firstly and MBA enumerates $\delta$ firstly. For the quantity of invalidated triangles, MBA only needs to invalidate each triangle in $\mathcal{G}^t$ once, while DBA does that for each $k$. Thus, the dominant time cost of MBA is $O(\sum_{k = 3}^{k_{max}} \sum_{(u,v) \in T_k} \min \{\deg(u),\deg(v)\} + |\Delta|)$, which means MBA is more efficient than DBA.

{\color{black}
\section{Index Maintenance}\label{sec:mainter}
% In the real world, temporal graphs typically exhibit dynamic change with temporal edge insertion and deletion. Owing to the irreversible nature of time, insertion operations are generally more frequent and significant than deletions in practical applications. Accordingly, this section proposes an efficient index maintenance algorithm that primarily focuses on the insertion of temporal edges.
Since temporal graphs are naturally evolving, we address dynamic index maintenance in this section. We have two reasonable assumptions. Firstly, the history should not be changed, which means that we do not consider edge deletion (note that, deletion can still be handled like insertion in the same way as our approach). Secondly, it is not necessary to restrict the order of timestamps for edge insertion, which makes our approach more flexible. As a result, an evolving temporal graph can be seen as a stream of edges with arbitrary timestamps. As new edges are inserted, the $k$-span of existing edges may change accordingly, incurring updates to both the TC-Index and the DC-Index. Instead of rebuilding the entire indexes from scratch, we propose a filter-and-verification algorithm to find the set of edges with changed $k$-spans, and only update the positions of these edges in the indexes.

In the following, we denote by $e^t_0 = (u, v, t)$ a temporal edge and by $e_0 = (u, v)$ its corresponding static edge. After inserting $e^t_0$ into a temporal graph $\mathcal{G}^t$, let $\mathcal{G}^t_+$ and $\mathcal{G}_+$ denote the updated temporal graph and its corresponding static graph, respectively.

\subsection{Filter of $k$}

In order to find the set of edges with changed $k$-spans, the first step is to identify a range of $k$ such that the $k$-spans of all edges are not changed for each $k$ that is out of the range. For that, we have the following observation.

\begin{theorem}\label{ins:theorem1} 
When a temporal edge $e^t_0 = (u, v, t)$ is inserted into $\mathcal{G}^t$, if $k > \trn(e_0,\mathcal{G}_+)$, the insertion of $e^t_0$ does not change the $k$-span of any edge $e$ in the $k$-truss of $\mathcal{G}$ ($T_k(\mathcal{G})$), namely, $k\text{-}\kspan(e, \mathcal{G}^t) = k\text{-}\kspan(e, \mathcal{G}^t_+)$. 
\end{theorem}

\begin{proof} 
When $k > \trn(e_0, \mathcal{G}_+)$, it has been proven in~\cite{huang2014trusscommunity} that $T_k(\mathcal{G})$ remains identical to that of $\mathcal{G}_+$ (i.e., $T_k(\mathcal{G}_+)$). Since $e_0$ does not belong to $T_k(\mathcal{G})$, the minimum time span of all triangles in $T_k(\mathcal{G})$ remains unchanged. Thus, the $k$-span of each edge in $T_k(\mathcal{G})$ also remains unaffected. 
\end{proof}

Intuitively, for values of $k$ greater than the trussness of $e_0$ in the updated graph, both the $k$-truss and the minimum time spans of its triangles in the original graph remain unchanged, so that the $k$-spans of all its edges will not be changed. Based on Theorem~\ref{ins:theorem1}, we only need to update the $k$-span of edges in $T_k(\mathcal{G})$ with $k \leq \trn(e_0, \mathcal{G}_+)$. 

% Any edge $e \in T_{k,\delta}$ must participate in at least $k-2$ triangles whose $\mathcal{K}$-weight is not greater than $\delta$. The $k$-span of an edge $e$ may decrease only if $e$ gains new supporting triangles with $\mathcal{K}$-weight less than $\delta$, or if existing supporting triangles experience a reduction in $\mathcal{K}$-weight from at least $\delta$ to below $\delta$.

% In the following, we explore how the insertion of a temporal edge $e^t = (u,v,t)$ affects the $k$-span for different cases.

\subsection{Filter of $k$-Span}

For each $k$ in the identified range, we aim to bound the values of $k$-span that could be changed. There are two cases that need to be discussed, respectively.

% \begin{lemma} \label{k-weight}
%     Given a temporal graph $\mathcal{G}^t$, an integer $k \geq 2$, and a triangle $\Delta \in \mathcal{G}$, let $\mathcal{K}_o$ and $\mathcal{K}_n$ denote the $k$-weight of $\Delta$ before and after inserting a timestamp,. For an edge $e' \in \mathcal{E}$, the $k$-span of $e'$ may decrease only if $\mathcal{K}_o < \mathcal{K}(e') \leq \mathcal{K}_n$ and one of following conditions holds: (i) $e' \in \Delta$ and $e'$ is connected to $\Delta$ through a series of triangles sharing common edges, each with $k$-weight not exceeding $\mathcal{K}_o$ and exceeding $\mathcal{K}_n$.
% \end{lemma}

% \noindent\textbf{Case 1: Timestamp Insertion.} Let $e^t = (u, v, t)$ be a temporal edge to be inserted into the temporal graph $\mathcal{G}^t$, and define the updated graph as $\mathcal{G}^t_+ = \mathcal{G}^t \cup e$. If the two vertices have already been connected (namely, $e = (u, v) \in \mathcal{E}$) it only needs to add the timestamp $t$ to $\tau(e)$ . This operation does not change the trussness of any edge and can only reduce the minimum time span of existing triangles containing $e$. Moreover, only some of such triangles lead to changes in the $k$-span of edges. The following lemma provides a criterion for identifying these triangles. 

\noindent\textbf{1) Timestamp Insertion.} When inserting a temporal edge $e^t_0$ into $\mathcal{G}^t$, if the underlying static edge $e_0$ already exists (namely, $e_0 \in \mathcal{E}(\mathcal{G})$), the insertion only needs to add the timestamp $t$ to $\tau(e_0)$. This operation does not change the trussness of any edge and only reduce the minimum time span of existing triangles containing $e_0$. Moreover, only some of such triangles lead to decreases in the $k$-span of edges. The following lemma provides a criterion for identifying these triangles. 

\begin{lemma}\label{lem:ctriangle} 
Given a temporal graph $\mathcal{G}^t$, a support threshold $k$ and a temporal edge $e^t_0$ with $e_0 \in \mathcal{E}(\mathcal{G})$ to be inserted, for each triangle $\Delta$ containing $e_0$, the $k$-span of any edge will not be decreased by the decrease of $\mts(\Delta)$, if $\mts(\Delta, \mathcal{G}^t) < \delta_m$ or $\mts(\Delta, \mathcal{G}^t_+) \geq \delta_m$, where $\delta_m = \max\{k\text{-}\kspan(e,\mathcal{G}^t) : e \in \Delta\}$.
\end{lemma}

\begin{proof}
When $\mts(\Delta, \mathcal{G}^t) < \delta_m$, for any edge $e \in \Delta$ with $k\text{-}\kspan(e) = \mts(\Delta, \mathcal{G}^t)$, its $k$-span cannot decrease, as this would require at least one triangle in $\mathcal{G}$ whose minimum time span drops from at least $\delta_m$ to below $\delta_m$, which is impossible. Moreover, for any $e \in \Delta$ with $k\text{-}\kspan(e) < \mts(\Delta, \mathcal{G}^t)$, its $k$-span remains unchanged, since $\Delta$ cannot be a valid triangle in its corresponding truss either before or after the insertion. Therefore, no edge in $\Delta$ has its $k$-span decreased, and thus no edge in $\mathcal{G}$ is affected by $\Delta$. The case for when $\mts(\Delta, \mathcal{G}^t) \geq \delta_m$ can be proved analogously.
\end{proof}

Thus, only the triangles that both contain $e_0$ and do not satisfy the conditions in Lemma~\ref{lem:ctriangle} need to be considered, which are denoted by $\mathcal{T}^{\Delta}_{e_0}$. For each triangle $\Delta \in \mathcal{T}^{\Delta}_{e_0}$, we compute a corresponding $k$-span interval $[\delta^-_{\Delta}, \delta^+_{\Delta}]$ $(\delta^-_{\Delta} <\delta^+_{\Delta})$ such that the set of edges in $\mathcal{G}$ with $k$-span outside this interval remains unchanged in $\mathcal{G}_+$ due to the change induced by $\Delta$. 

% Note that there cannot exist a $k$-span interval $[\delta^-_{\Delta}, \delta^+_{\Delta}]$ with $\delta^-_{\Delta} = \delta^+_{\Delta}$ as any decrease in the $k$-span of an edge affect at least two distinct $k$-span values, causing the set of edges with those values to differ between $\mathcal{G}$ and $\mathcal{G}_+$

% \begin{lemma}\label{ins:case1:lemma1} 
% Given a temporal graph $\mathcal{G}^t$, an integer $k \geq 2$, for each triangle $\Delta \in \mathcal{T}_{\Delta}$, $\delta_{\Delta}^+ = \max\{k\text{-}\kspan(e,\mathcal{G}^t) | e \in \Delta \}$. 
% \end{lemma}

% To find a proper $\underline{\delta}$, we introduce the concept of $k$-weight.

% \begin{definition}[$k$-Weight]\label{def:kweight} 
% Given a temporal graph $\mathcal{G}^t$ and an integer $k \geq 2$, the k-weight of $\Delta \in \mathcal{G}$ is an integer $\lambda = w(\Delta,k,\mathcal{G})$, such that (i) $\Delta$ is a valid triangle in $(k, \lambda)$-truss and (ii) $\Delta$ is not a valid triangle in $(k,\lambda')$-truss with $\lambda' < \lambda$. When the context is clear, we replace $w(\Delta, k ,\mathcal{G}^t)$ by $kw(\Delta)$.
% \end{definition}

% $w(\Delta, k, \mathcal{G}^t)$, is defined as $w(\Delta,k,\mathcal{G}^t) = \min \{\delta: e \in \Delta, k\text{-}\kspan(e, \mathcal{G}^t) \geq \delta , \mts(\Delta, \mathcal{G}^t) \geq \delta\}$. When the context is clear, we replace $w(\Delta, k ,\mathcal{G}^t)$ by $kk\text{-}\krank(\Delta)$. 

% Any edge $e$ that belongs to a $(k,\delta)$-truss must participate in at least $k - 2$ triangles whose $k$-weight is not greater than $\delta$ and
For each $\Delta \in \mathcal{T}^{\Delta}_{e_0}$, it is evident that $\delta^+_{\Delta}$ cannot exceed the maximum $k$-span of the edges in triangle $\Delta$, i.e., $\delta^+_{\Delta} = \max\{k\text{-}\kspan(e,\mathcal{G}^t) : e \in \Delta \}$. To determine a proper $\delta^-_{\Delta}$ , the key is to identify the minimum $\delta$ such that $\Delta$ can be a valid triangle in a $(k,\delta)$-truss. We define $\hat{\delta} = \max\{\delta^+_{\Delta} : \Delta \in \mathcal{T}^{\Delta}_{e_0}\}$, which represents the maximum $k$-span among all edges potentially affected by the insertion of a temporal edge. Then for any edge $e \in \Delta \in \mathcal{T}^{\Delta}_{e_0}$, the lower bound of its $k$-span is given by $\underline{\delta}(e) = \min \{\delta_a : |\{e, e', e''\}| \geq k - 2, k\text{-}\kspan(e', \mathcal{G}^t)  \leq \hat{\delta},  
k\text{-}\kspan(e'', \mathcal{G}^t) \leq \hat{\delta}, \mts(\Delta, \mathcal{G}^t_+) \leq \delta_a\}$, which implies that $e$ can participate in at most $k - 2$ valid triangles in $(k,\delta_a)$-truss. Furthermore, $\underline{\underline{\delta}}(e) \,{=}\, \min\{\delta_b : |\{e, e', e''\}| \geq k - 2, \underline{\delta}(e') \leq \delta_b,  \underline{\delta}(e'') \leq \delta_b,  \mts(\Delta, \mathcal{G}^t_+) \leq \delta_b\}$ serves as a tighter lower bound for the $k$-span of $e$, as it incorporates an extra constraint regarding $\underline{\delta}$ values. Hence, triangle $\Delta$ only exist in $(k, \delta^*)$-truss where $\delta^* \,{\geq}\, \max\{\mts(\Delta, \mathcal{G}^t_+), \max\{\underline{\underline{\delta}}(e), e \in \Delta\}\}$.

\begin{lemma}\label{ins:case1:lemma4} 
Given a triangle $\Delta \in \mathcal{T}^{\Delta}_{e_0}$ and a support threshold $k$, let $\delta^-_{\Delta} = \max\{\mts(\Delta, \mathcal{G}^t_+), \max\{\underline{\underline{\delta}}(e): e \in \Delta\}\}$ and $\delta^+_{\Delta} = \max\{k\text{-}\kspan(e,\mathcal{G}^t) : e \in \Delta \}$, for any integer $\delta'$ such that $\delta' < \delta^-_{\Delta}$ or $\delta' > \delta^+_{\Delta}$, the set of edges in $\mathcal{G}$ with $k$-span equal to $\delta'$ remains unaffected in $\mathcal{G}_+$ by the change induced by $\Delta$.
\end{lemma}

\begin{proof}
    We first consider the case $\delta' > \delta^+_{\Delta}$. For any edge $e \in \mathcal{E}(\mathcal{G})$ with $k\text{-}\kspan(e) > \delta^+_{\Delta}$, suppose that its $k$-span is decreased. Then there must exist an edge $e'$ sharing a triangle with $e$ whose $k$-span also drops from at least $\delta^+_{\Delta}$ to below $\delta^+_{\Delta}$. By applying this recursively through a series of triangles sharing common edges, we eventually reach an edge $e^*$ of $\Delta$ with $k\text{-}\kspan(e^*) > \delta^+_{\Delta}$ and its $k$-span is decreased due to the reduction in the minimum time span of $\Delta$. However, by the definition of $\delta^+_{\Delta}$, no such edge $e^*$ exists in $\Delta$, which leads to a contradiction. Moreover, the reduction in the minimum time span of $\Delta$ cannot cause any edge to increase its $k$-span. Thus, the set of edges with $k$-span equal to $\delta'$ keeps same.
    
    Then we prove the case $\delta' < \delta^+_{\Delta}$. For any edge $e \in \Delta$, since $k\text{-}\kspan(e) <= \hat{\delta}$, it follows that $k\text{-}\kspan(e) \geq \underline{\delta}(e) \geq \underline{\underline{\delta}}(e)$ and $\max\{\mts(\Delta, \mathcal{G}^t), \max\{k\text{-}\kspan(e) : e \in \Delta\}\} \geq \delta^-_{\Delta}$. Therefore, the $(k,\delta')$-truss in $\mathcal{G}$ does not contain $\Delta$. Similarly, since $\delta' < \delta^-_{\Delta}$, $\Delta$ is also excluded from the $(k,\delta')$-truss in $\mathcal{G}_+$. Consequently, the $(k,\delta')$-truss remains unchanged, and thus the set of edges with $k$-span equal to $\delta'$ keeps same.
\end{proof}

Therefore, after inserting a temporal edge $e^t_0$, the affected range of $k$-span values can be identified by $\mathcal{R}(e^t_0) = \bigcup_{\Delta \in \mathcal{T}^{\Delta}_{e_0}} [\delta^-_{\Delta}, \delta^+_{\Delta}]$, which can be simply merged to a set of disjoint intervals.

\begin{figure}[t!] 
    \centering 
    \includegraphics[width=0.6\linewidth]{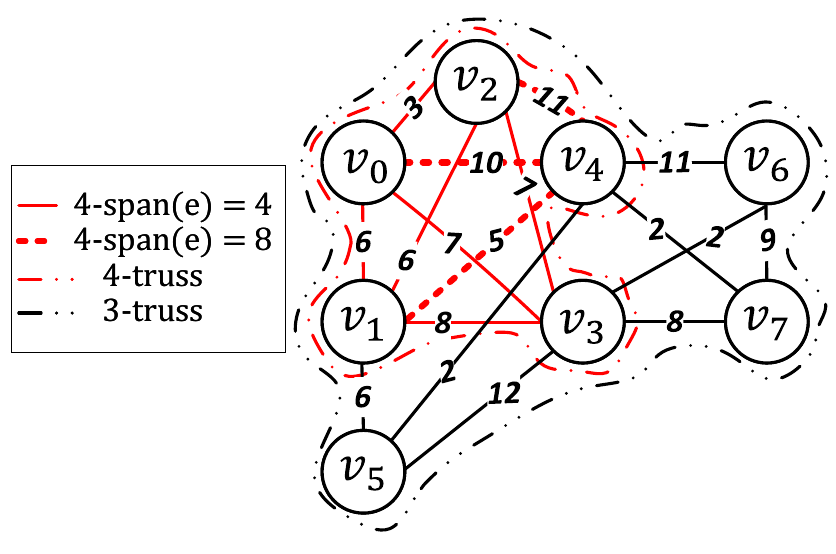} 
    \caption{\textcolor{black}{An example of estimating the k-span upper bound for a new edge.}} 
    \label{fig:tg2} 
    \vspace{-0.5cm}
\end{figure}

\noindent\textbf{2) Edge Insertion.} When a temporal edge $e^t_0$ is inserted into $\mathcal{G}^t$ with $e_0 \notin \mathcal{E}(\mathcal{G})$, the static edge $e_0$ is added to $\mathcal{E}(\mathcal{G})$ and its timestamp set initialized as $\tau(e_0) = {t}$. According to~\cite{huang2014trusscommunity}, the trussness of $e_0$ in updated graph $\mathcal{G}^t_+$ can be estimated within the interval $[k_a, k_b]$, where $k_a = \max\{k: |\{e_0, e', e''\}| \geq k - 2, trn(e', \mathcal{G}) \geq k, trn(e'', \mathcal{G}) \geq k)$, $k_b = \max\{k: |\{e_0,e', e''\}| \geq k - 2, trn(e', \mathcal{G}) \geq k - 1, trn(e'', \mathcal{G}) \geq k - 1)$ and $k_b - k_a \geq 1$. Assuming that the trussness of $e_0$ in $\mathcal{G}$ is $k_a$, for $e \in \mathcal{E}(\mathcal{G}) \cup e_0$ with $trn(e, \mathcal{G}) < trn(e_0, \mathcal{G}_+)$ may have their trussness increased by 1. As a result, new edges may be introduced into $T_k(\mathcal{G})$. Since these edges are not originally in $T_k(\mathcal{G})$ and their $k$-spans are undefined, for each such edge $e$, we estimate an upper bound on its $k$-span, denoted by $\overline{\delta}(e)$, before obtaining the $k$-span interval $[\delta^-_{\Delta}, \delta^+_{\Delta}]$ for each triangle $\Delta$ containing $e_0$.
% , and the resulting temporal graph is denoted as $\mathcal{G}^t_+ = \mathcal{G}^t \cup e^t$. According to~\cite{huang2014trusscommunity}, inserting edge $e$ may alter $T_k(\mathcal{G})$. Specifically, for $k \in [3, \trn(e, \mathcal{G}_{+})]$, edge $e$ must be added to every $T_k(\mathcal{G})$. Similarly, any edge $e'$ with $\trn(e', \mathcal{G}) \,{=}\, k - 1$ and $\trn(e', \mathcal{G}_+) \,{=}\, k$ must be added to $T_k(\mathcal{G})$. As these edges are not originally in $T_k(\mathcal{G})$ and their $k$-spans are undefined, we estimate an upper bound for their $k$-span in $\mathcal{G}^t$ to construct a tighter $k$-span window $[\delta^-, \delta^+]$ for each triangle containing $e$.

First, for $k \in [3, k_a]$, the edge $e_0$ is added into $T_k(\mathcal{G})$, and the upper bound of its $k$-span can be estimated by: $\overline{\delta}(e_0) = \min\{\delta : |\{e_0,e',e''\}| \geq k - 2, k\text{-}\kspan(e',\mathcal{G}^t) \leq \delta, k\text{-}\kspan(e'',\mathcal{G}^t) \leq \delta, \mts(\Delta, \mathcal{G}^t_+) \leq \delta\}$. Next, we estimate an upper bound on the $k$-span of edges whose trussness increases by one after the insertion of $e_0$.

\begin{definition}[$\mathcal{L}_{\mathcal{E}_k}$] 
When a temporal edge $e^t_0$ with $e_0 \notin \mathcal{E}(\mathcal{G})$ is inserted into $\mathcal{G}^t$, for each $k \in [3, trn(e_0, \mathcal{G}_+)]$, we define the set $\mathcal{L}_{\mathcal{E}_k} = \{\mathcal{L}_{\mathcal{E}_k}^1, \mathcal{L}_{\mathcal{E}_k}^2, \cdots, \mathcal{L}_{\mathcal{E}_k}^i\}$, where $0 \leq i$, satisfying: 
\begin{itemize}
\item[(1)] $\forall \, e \in \mathcal{L}_{\mathcal{E}_k}$, $\trn(e, \mathcal{G}_+) = k$ and $\trn(e, \mathcal{G}) = k - 1$; 

% \item[(2)] $\forall \, e_1 \in \mathcal{L}_{\mathcal{E}_k}^j$ and $\forall \, e_2 \in \mathcal{L}_{\mathcal{E}_k}^l$, if $j \neq l$, there does not exist a $\Delta$ in $\mathcal{G}_+$ with $L(\Delta,\mathcal{G}) = k$  containing both $e_1$ and $e_2$. 
\item[(2)] For any $\mathcal{L}_{\mathcal{E}_k}^j$, any two different edges in it can be connected through a series of adjacent $(k - 1)$- triangles formed by edges in $\mathcal{E}(\mathcal{G}) \cup e_0$;

\item[(3)] For $\forall e_j \in \mathcal{L}_{\mathcal{E}_k}^j$, $\forall e_l \in \mathcal{L}_{\mathcal{E}_k}^l$, $j \neq l$, there does not exist any $(k - 1)$-triangle formed by the edges $\mathcal{E}(\mathcal{G}) \cup e_0$ that contains both $e_j$ and $e_l$.
\end{itemize}

\end{definition}

For each $\mathcal{L}_{\mathcal{E}_k}^i$ within $\mathcal{L}_{\mathcal{E}_k}$, we can independently estimate an upper bound $\overline{\delta}(\mathcal{L}_{\mathcal{E}_k}^i)$ on the $k$-span of the internal edges.

\begin{definition}[$\overline{\delta}(\mathcal{L}^i_{\mathcal{E}_k})$] 
For each $k \in [3, trn(e_0, \mathcal{G}_+)]$, let $e'$ be any edge in $\mathcal{L}_{\mathcal{E}_k}^i$ and $e''$ be any edge not in $\mathcal{L}_{\mathcal{E}_k}^i$, the upper bound of the $k$-span for edges in $\mathcal{L}_{\mathcal{E}_k}^i$ is defined as: 
    \begin{equation} 
        \begin{aligned} 
        &\overline{\delta}(\mathcal{L}_{\mathcal{E}_k}^i) = \max\{t_1,t_2\} \\ 
        &t_1 = \max\{ \mts(\Delta, \mathcal{G}^t_+): e' \in \Delta,L(\Delta,\mathcal{G}_+) = k\}\\
        &t_2 = \max\{k\text{-}\kspan(e'', \mathcal{G}^t) : e', e'' \in \Delta , \trn(e'',\mathcal{G}) \geq k\} \\
        \end{aligned} 
    \end{equation} 
\end{definition}

\begin{lemma}\label{ins:case2:lemma1} 
For a given $k \in [3, trn(e_0, \mathcal{G}_+)]$ and any edge $e' \in \mathcal{L}^i_{\mathcal{E}_k}$, $\overline{\delta}(\mathcal{L}^i_{\mathcal{E}_k})$ serves as an upper bound on the $k$-span of $e'$, i.e., $\overline{\delta}(e') = \overline{\delta}(\mathcal{L}^i_{\mathcal{E}_k})$.  \end{lemma}
% it holds that $k\text{-}\kspan(e', \mathcal{G}^t_+) \leq \overline{\delta}(\mathcal{L}^i_{\mathcal{E}_k})$.
\begin{proof} 
For $\forall\, e' \,{\in}\, \mathcal{L}_{\mathcal{E}_k}^i$, by the definition of $t_1$, the edge $e'$ must participate in at least $k \,{-}\, 2$ triangles whose minimum time spans are less than $t_1$ and for which $L(\Delta,\mathcal{G}_+) \,{=}\, k$. Consider any $\Delta$ containing $e'$ and let $e^*$ denote any other edge in $\Delta$. There are two possible cases. On one hand, if $e^* \,{\notin}\, \mathcal{L}_{\mathcal{E}_k}^i$ and $\trn(e^*,\mathcal{G}) \geq k$, then we have $k\text{-}\kspan(e^*,\mathcal{G}^t_+) \,{\leq}\, t_2$. On the other hand, if $e^* \,{\in}\, \mathcal{L}_{\mathcal{E}_k}^i$, then $e^*$ satisfies the same triangle support condition as $e'$. By recursively applying this reasoning, it follows that $e'$ participates in at least $k-2$ valid triangles within $(k,\max\{t_1, t_2\})$-truss. Thus, $k\text{-}\kspan(e',\mathcal{G}^t_+) \,{\leq}\, \overline{\delta}(\mathcal{L}^i_{\mathcal{E}_k})$ and $\overline{\delta}(\mathcal{L}^i_{\mathcal{E}_k})$ is an upper bound for the $k$-span of $e'$ .
\end{proof}

% Specifically, maintaining an inserted temporal edge $e^t$ can be reduced to processing the associated triangle set $\mathcal{S}_{\Delta}$. During this process, a series of overlapping $k$-span time windows are identified. By merging these overlapping windows, we obtain a set of non-overlapping $k$-span windows, denoted as $\text{Win}(e^t)$. Edges whose $k$-span lies outside this window set remain unaffected.
% No triangle $\Delta$ formed by edges in $\mathcal{E}(\mathcal{G}) \cup (v_3, v_4)$ with $L(\Delta) = 3$ contains edges from both sets, ensuring that their trussness is maintained independently. 

\begin{example}
     When inserting the temporal edge $(v_3, v_4, 2)$ into the temporal graph in Fig~\ref{fig:tg2}, accroding to the definition of $k_a$ and $k_b$, the trussness of edge $(v_3, v_4)$ is estimated within the range $[4, 5]$. For $k = 4$, $\overline{\delta}((v_3,v_4))$ is computed as $8$, supported by triangles $\{v_0,v_3,v_4\}$ and $\{v_1,v_3,v_4\}$. Moreover, the insertion triggers an increase in trussness to $4$ for two disjoint edge sets: $\mathcal{L}_{\mathcal{E}_4}^1  \,{=}\, \{(v_1,v_5),(v_3,v_5), (v_4,v_5)\}$ and $\mathcal{L}_{\mathcal{E}_4}^2  \,{=}\, \{(v_3,v_6),\allowbreak(v_3,v_7),(v_4,v_6), (v_4,v_7), (v_6, v_7)\}$. For $\mathcal{L}_{\mathcal{E}_k}^1$, the triangles satisfying the conditions in $t_1$ include $\{\{ v_1, v_3, v_5\}, \{ v_1, v_4, v_5\}, \{ v_3, v_4, v_5\}\}$, and $t_1$ is computed as $10$. Then, the edges satisfying the conditions in $t_2$ are $\{(v_1,v_{3}), (v_1,v_{4}), (v_3, v_4)\}$, and $t_1$ is computed as $8$. Thus, $\overline{\delta}(\mathcal{L}^1_{\mathcal{E}_4}) = \max\{10,8\} = 10$. Similarly, $\overline{\delta}(\mathcal{L}^2_{\mathcal{E}_4}) {=} \max\{8,9\} = 9$. For $k = 5$, the edges between the vertices $v_0, v_1, v_2, v_3, v_{4}$ pairwise form $\mathcal{L}_{\mathcal{E}_5}^1$, and $\overline{\delta}(\mathcal{L}_{\mathcal{E}_5}^1)$ is computed as $8$.   
     
     % = \max\{\mts(), \mts(),  \mts(), \mts()\} {=} \max\{6, 6, 4{,} 10\} {=} \allowbreak 10$ and $t_2 = \max\{\mathcal{K}((v_1,v_{3})), \mathcal{K}((v_1,v_{4})),  \mathcal{K}((v_{1}, v_{5}))\} = \max\{8, 4, 8\} {=} 8$, thus $\overline{\mathcal{K}}(\mathcal{L}^1_{\mathcal{E}_4}) {=} \max\{10,8,8\} {=} 10$. Similarly, $\overline{\mathcal{K}}(\mathcal{L}^2_{\mathcal{E}_4}) {=} \max\{8,9,-1\} {=} 9$. When $k {=} 5$, the edges between  the vertices $v_0{,} v_1{,} v_2{,} v_3{,} v_{4}$ pairwise form $\mathcal{L}_{\mathcal{E}_5}^1$, and we compute $t_1 = 8$, $t_2 = -1$, hence $\overline{\mathcal{K}}(\mathcal{L}_{\mathcal{E}_5}^1) = \max\{8, -1, -1\} = 8$. Note that we set $t_1$ and $t_2$ to $-1$ initially. 
\end{example}

For all newly inserted edges in $T_k(\mathcal{G})$, we assign their upper bounds of $k$-span to be their $k$-span values in the original temporal graph. Accordingly, the edge insertion can be equivalently transformed into a timestamp insertion maintenance problem, namely, reducing the minimum time span of certain triangles from $\infty$ to $\mts(\Delta, \mathcal{G}^t_+)$.

\begin{algorithm}[t!] 
    \SetAlgoNoEnd \SetKwInput{Input}{Input} 
    \SetKwInput{Output}{Output} 
    \SetAlgoLined 
    \SetAlgoVlined 
    \color{black}
    \Input{a temporal graph $\mathcal{G}^t$, the inserted temporal edge $e^t_0$, a support threshold $k$ and a $k$-span interval $[\delta^-, \delta^+]$} 
    \Output{The potentially affected edge $T^s_{k,\delta^+}$ and the local $\delta$-triangle list $\mathcal{S}^{\Delta}$} 
    \caption{Get Affected Subgraph (GAS)}\label{alg:localsubgraph} 
    $\mathcal{S}^{\Delta} \leftarrow \varnothing$; $Q \leftarrow \varnothing$; $visited\leftarrow \varnothing$\; 
    \If{\textnormal{$k\text{-}\kspan(e_0) \leq \delta^+$ and $k\text{-}\kspan(e_0) \geq \delta^-$}}{ 
        $Q$.push($e_0$); $visited\leftarrow visited\ \cup e_0$; 
    } 
    \For{\textnormal{each $\Delta$ containing $e_0$ in $\mathcal{G} \cup e_0$}}{
        \If{$k\text{-}\krank(\Delta) > \delta^+$}{continue;}  
        \For{\textnormal{each other $e$ of $\Delta$}}{
            \If{\textnormal{$k\text{-}\kspan(e) \leq \delta^+$ and $k\text{-}\kspan(e) \geq \delta^-$}}{ 
                $Q$.push($e$); $visited\leftarrow visited\ \cup e$; 
            } 
        } 
    } 
    \While{$Q \neq \varnothing$}{ 
        $e \leftarrow Q$.pop(); $\sup[e]\leftarrow 0$\; 
        \For{\textnormal{each $\Delta$ containing $e$ in $\mathcal{G} \cup e_0$}}{ 
            \If{$k\text{-}\krank(\Delta) > \delta^+$}{continue;} 
            $\mathcal{S}^{\Delta}_{\mts(\Delta)}\leftarrow\mathcal{S}^{\Delta}_{\mts(\Delta)}\cup\Delta$\; 
            $\sup[e]\leftarrow \sup[e] + 1$\; 
            \For{\textnormal{each other $e'$ of $\Delta$}}{ 
                \If{$e''\notin visited$}{ 
                    \If{$k\text{-}\kspan(e') \geq \delta^-$}{ 
                        $Q$.push($e'$); 
                    } \Else{ 
                        $s[e'] \leftarrow \infty$; 
                    } 
                    $visited\leftarrow visited\ \cup e'$; 
                } 
            } 
        } 
    } 
    \Return $T^s_{k,\delta^+}$ and $\mathcal{S}^{\Delta}$\; 
\end{algorithm}

\begin{algorithm}[t!]
    \SetAlgoNoEnd
    \SetKwInput{Input}{Input}  
    \SetKwComment{Comment}{// }{}
    % Set the Input
    \SetKwInput{Output}{Output} % set the Output
    % \SetAlgoLined
    % \SetAlgoVlined
    \color{black}
    \Input{a temporal graph $\mathcal{G}^t$ and the inserted temporal edge $e^t_0$}
    \Output{updated $k$-spans of edges for each $k$}
    \caption{Filter-and-Verification}\label{alg:insertmaintenance}
    % $\mathcal{AFF} \leftarrow \varnothing$\Comment*[r]{record updated $k$-span} 
    $\mathcal{G}^t_+ \leftarrow \mathcal{G}^t \cup e^t_0$\;
    \For{\textnormal{$k \leftarrow \trn(e_0, \mathcal{G}_+)$ until $3$}}{
        \If{\textnormal{$e_0 \notin \mathcal{E}$}}{
            % $k\text{-}\kspan(e) \leftarrow \overline{\delta}(e,k,\mathcal{G}^t_+)$\;
             identify the edge set $T_k(\mathcal{G}_+) \setminus T_k(\mathcal{G})$ by truss maintenance\;
            \For{\textnormal{each edge $e \in T_k(\mathcal{G}_+) \setminus T_k(\mathcal{G})$}}{
                $k\text{-}\kspan(e) \leftarrow \overline{\delta}(e)$; 
            } 
        }
        compute $\mathcal{T}^{\Delta}_{e_0}$ using the proposed Lemma~\ref{lem:ctriangle}\;
        % \If{$k = 3$}{
        %     update the $k$-span of the edges in $\mathcal{T}^{\Delta}_{e_0}$\;
        %     continue\;
        % }
        \For{\textnormal{each $\Delta \in \mathcal{T}^{\Delta}_{e_0}$}}{
            compute $[\delta^-_{\Delta},\delta^+_{\Delta}]$\;
             $\mathcal{R}(e^t_0) \leftarrow \mathcal{R}(e^t_0) \cup [\delta^-_{\Delta}, \delta^+_{\Delta}]$\;
            % \If{$\delta^+_{\Delta} \neq \delta^-_{\Delta}$}{
               
            % }
        } 
        merge overlapping intervals in $\mathcal{R}(e^t_0)$\;
        \For{\textnormal{each $[\delta^-, \delta^+] \in \mathcal{R}(e^t)$}}{
            $T^s_{k,\delta^+},\mathcal{S}^{\Delta} \leftarrow \text{GAS}(\mathcal{G}^t, e^t_0, k, [\delta^-, \delta^+])$\;
            \For{\textnormal{$\delta \leftarrow \delta^+$ until $\delta^-$}}{
                $T^s_{k,\delta - 1}, R \leftarrow$ \dch{$T^s_{k,\delta}, \mathcal{S}^{\Delta}$}\;
                \For{\textnormal{each $e \in R$}}{
                    \If{$k\text{-}\kspan(e) \neq \delta$}{
                        update $k\text{-}\kspan(e)$ to $ \delta$\;
                    }
                }
                % $\mathcal{G}_{\textnormal{a}} \leftarrow \mathcal{G}_{\textnormal{a}}'$\;
            } 
            
        }
    }
    \Return all updated $k$-span values of the edges in $\mathcal{E}(\mathcal{G})$\;
\end{algorithm}

\subsection{Filter of Edge}
% Similar to the weight of triangle defined in ~\cite{huang2014trusscommunity}, which is the maximum k value for a triangle to belong to a k-truss,for each $k$-span window $[\delta^-, \delta^+]$ in $\mathcal{R}(e^t)$, 
For each $k$-span interval $[\delta^-, \delta^+]$ in $\mathcal{R}(e^t_0)$, the potentially affected edges with $k$-span within this interval can be further filtered since our $k$-span maintenance exhibits a triangle propagation property similar to that of truss maintenance~\cite{huang2014trusscommunity}. We define the $k$-rank of triangle $\Delta$, denoted by $k\text{-}\krank(\Delta)$, as the minimum $\delta$ value such that triangle can be included in a $(k,\delta)$-truss for the given $k$. Based on this definition, after inserting a temporal edge $e^t_0$, a potentially affected edge $e$ needs to satisfy two conditions: 1) $k\text{-}\kspan(e) \in [\delta^-, \delta^+] \in \mathcal{R}(e^t_0)$ and 2) $e$, $e_0$ and another $e'$ form a new triangle with $k\text{-}\krank(\Delta) \in [\delta^-, \delta^+]$ or $e$ is connected with $e_0$ through a series of adjacent triangles, each $\Delta$ of which has $k\text{-}\krank(\Delta) \in [\delta^-, \delta^+]$. 

Due to the connectivity of the triangles in the second condition, for each $k$-span interval $[\delta^-, \delta^+]$, we adopt a local search method to identify all potential affected edges, which form a subset of the $T_{k,\delta^+}$ of $\mathcal{G}$, denoted as $T_{k,\delta^+}^s$. During this process, we construct the $\delta$-triangle list $\mathcal{S}^{\Delta}$ of $T_{k,\delta^+}^s$ and count the potentially valid triangles associated with each candidate edge. Algorithm~\ref{alg:localsubgraph} provides the pseudocode for this procedure. For a given $k$ and a $k$-span interval $[\delta^-, \delta^+]$, we use breadth-first search to locate all edges with $k$-span values within $[\delta^-, \delta^+]$ that are connected to edge $e_0$ via a sequence of triangles whose $k$-ranks are not greater than $\delta^+$. (Line 2-23) For each such edge $e$, we consider all triangles with $k$-rank not exceeding $\delta^+$ as its potentially valid triangles, set the number of these triangles as the support of $e$, and add them to $\delta$-triangle list $\mathcal{S}^{\Delta}$ (lines 10-23). If an edge with a $k$-span less than $\delta^-$ is encountered, the search terminates for that branch, and its support is set to $\infty$ (lines 19-22).

\subsection{Verification of Edge $k$-Span}
Algorithm~\ref{alg:insertmaintenance} provides the pseudocode of the complete filter-and-verification Procedure. The process iterates over $k$ from $3$ to $\trn(e_0, \mathcal{G}_+)]$ (Line 1). For each possible $k$, the $k$-span filtering process considers both timestamp insertion and edge insertion. In the case of edge insertion, we employ a modified version of the trussness maintenance algorithm~\cite{huang2014trusscommunity} to identify edges that need to be inserted into $T_k(\mathcal{G})$ and assign their initial $k$-span values, thereby transforming the problem into the timestamps insertion case (lines 3–6). For the timestamp insertion case, we compute the filtered $k$-span intervals $\mathcal{R}(e^t_0)$ according to the proposed lemma(lines 7–11). For each $k$-span interval $[\delta^-, \delta^+]$ in $\mathcal{R}(e^t_0)$, Algorithm~\ref{alg:localsubgraph} is performed to extract the affected local subgraph $T^s_{k,\delta^+}$, and on $T^s_{k,\delta^+}$, we perform the decomposition function \dch{} in DBA from $\delta^+$ to $\delta^-$ to accurately verify the $k$-span values of all edges in $T^s_{k,\delta^+}$. (lines 12–18). 

With all updated $k$-span values of the edges, we can update TC-Index or DC-Index by changing the positions of the edges. Since the index updating is straightforward, the details are omitted.

}

\section{Experiments}\label{sec:expr}
In this section, we conduct extensive experiments to evaluate the proposed approaches on a Linux machine with Intel Xeon 3.5GHz CPU and 128GB RAM. All algorithms are implemented in C++ 11 and compiled by g++ with O3 optimization.

\subsection{Dataset and Empirical Study}

)We evaluate our proposed approaches on eight publicly available real-world temporal graphs from~\cite{lesko2014snap} and~\cite{kuneg2013konect}. Table~\ref{tab:dataset} summarizes their statistics, where $|\mathcal{V}|$ denotes the number of vertices, $|\mathcal{E}|$ the number of static edges, $n$ the number of distinct timestamps, $\overline{|\tau|}$ the average number of timestamps per edge, $|\Delta|$ the number of triangles, $k_{max}$ the maximum trussness of edges, and $\delta_{max}$ the maximum minimum time span of triangles. The datasets vary widely in scale, with $|\mathcal{E}|$ ranging from 16K to 36M. However, regardless of graph size, $k_{max}$ remains within a narrow range due to the strict static cohesion of $k$-truss subgraphs.

There are two observations that support our motivation of studying $(k,\delta)$-truss on the datasets. Firstly, by comparing $n$ and $\delta_{max}$, we can see there are indeed triangles with minimum time span as long as the duration of whole graph, which are not cohesive in terms of time. Moreover, we conduct an empirical study on these datasets. Fig~\ref{fig:estudy} illustrates the distribution of triangle counts on minimum time span for four of them. We can see that, although the triangles with longer minimum time span are less, the distribution does not have a typical long tail. In contrast, the counts are not dropping that fast. Thus, $\delta$ is effective to constrain the structure of $(k,\delta)$-truss.
\begin{figure}[t!]
    \centering
    \subfloat[Mathoverflow.]{
        \includegraphics[width=0.43\linewidth]{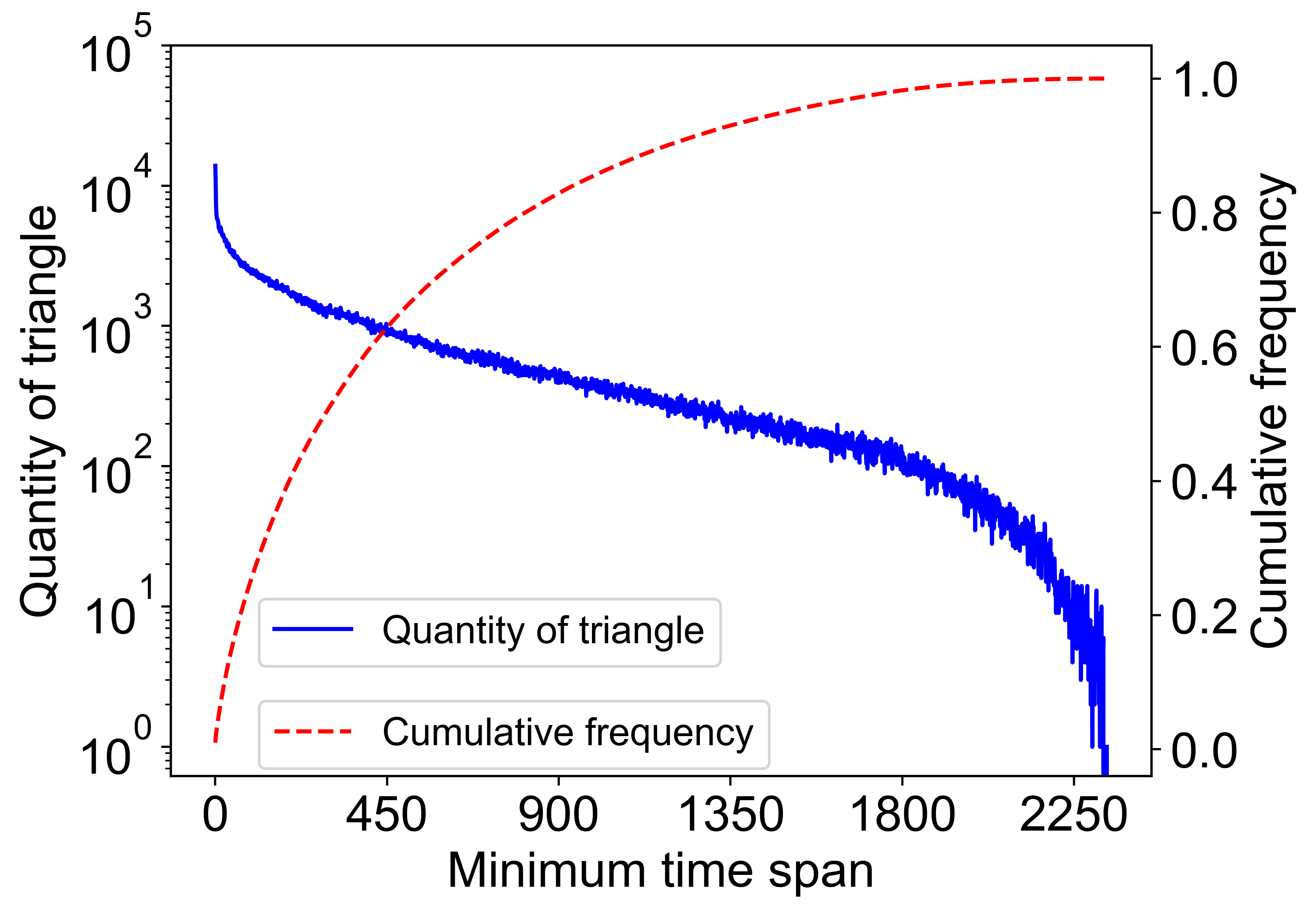}}
        \label{fig:triangleofmath}
    \subfloat[Superuser.]{
        \includegraphics[width=0.43\linewidth]{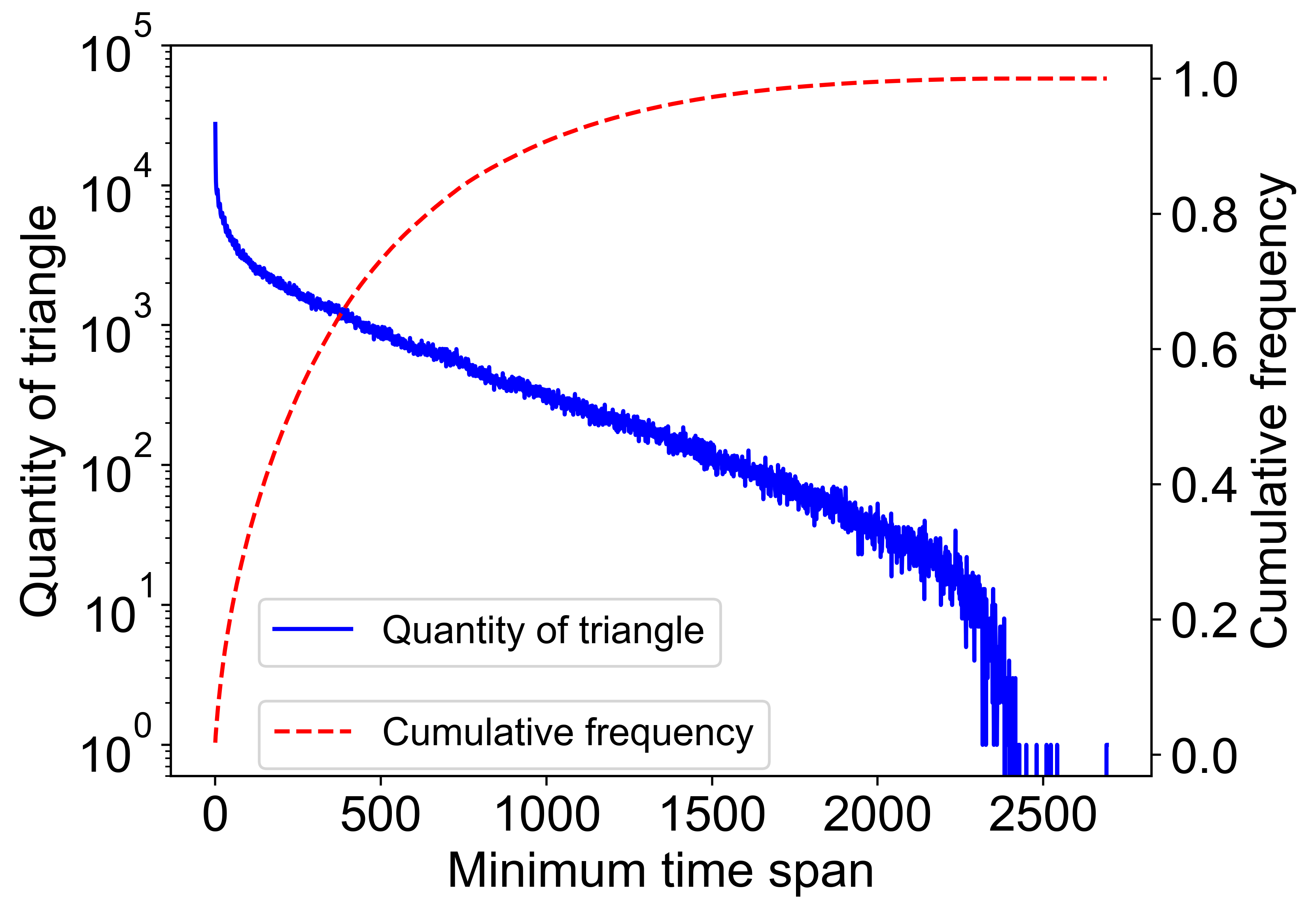}}
    \\
    \subfloat[Wikitalk.]{
        \includegraphics[width=0.43\linewidth]{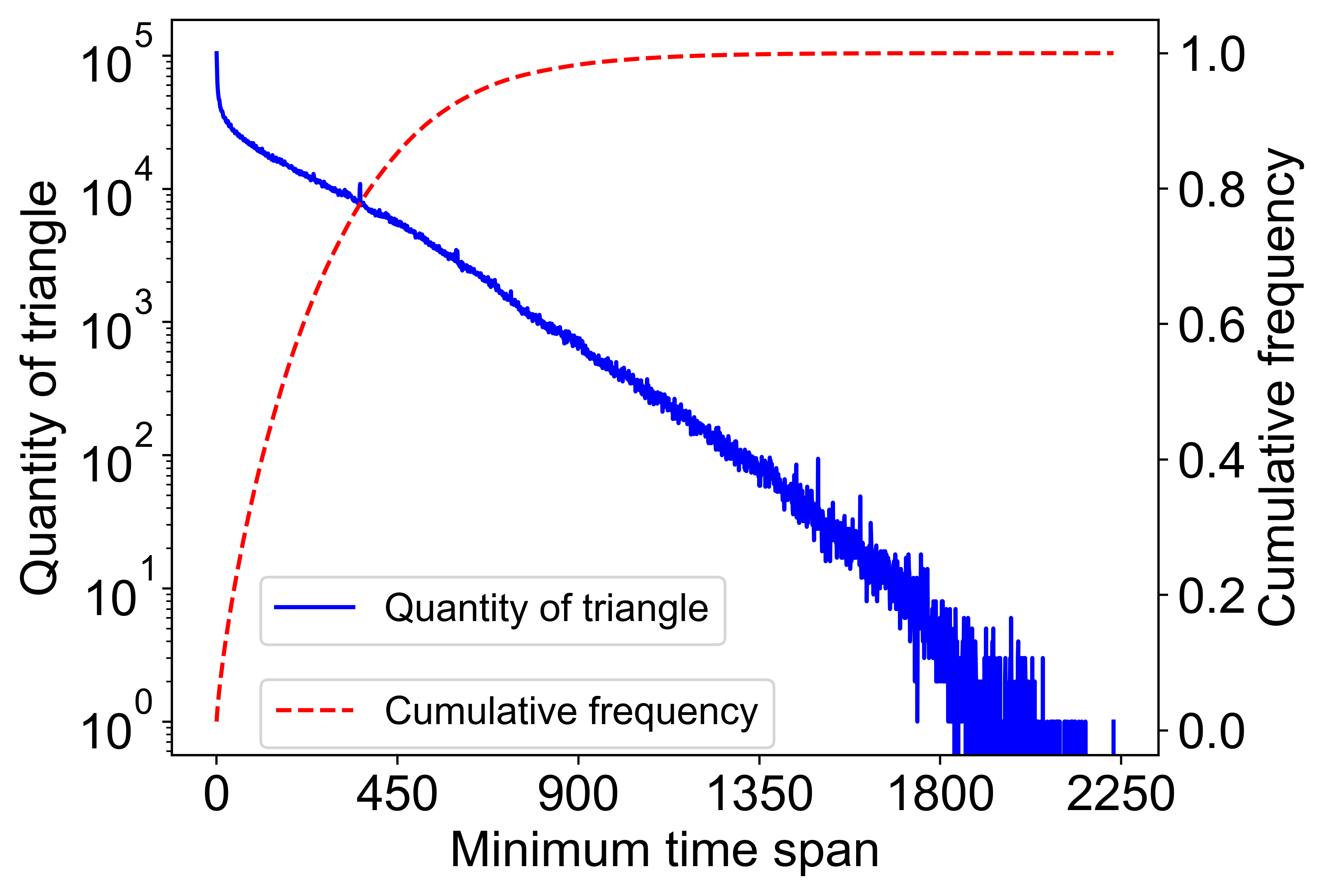}}
    \subfloat[Youtube.]{
        \includegraphics[width=0.43\linewidth]{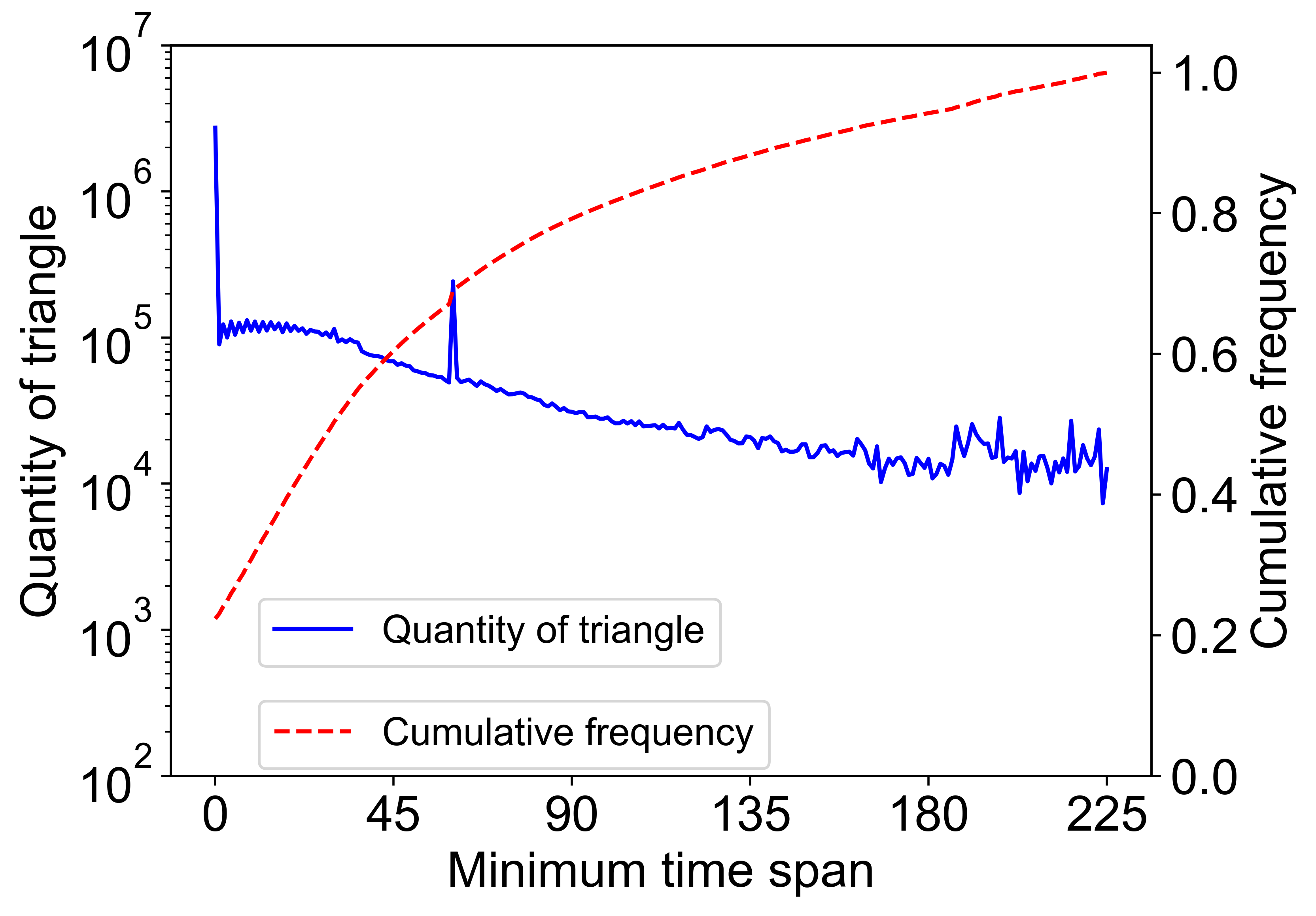}}
    \caption{Distribution of triangle counts on minimum time span.}\label{fig:estudy}
\end{figure}

\begin{figure}[t!]
    \centering
    \includegraphics[width=0.8\linewidth]{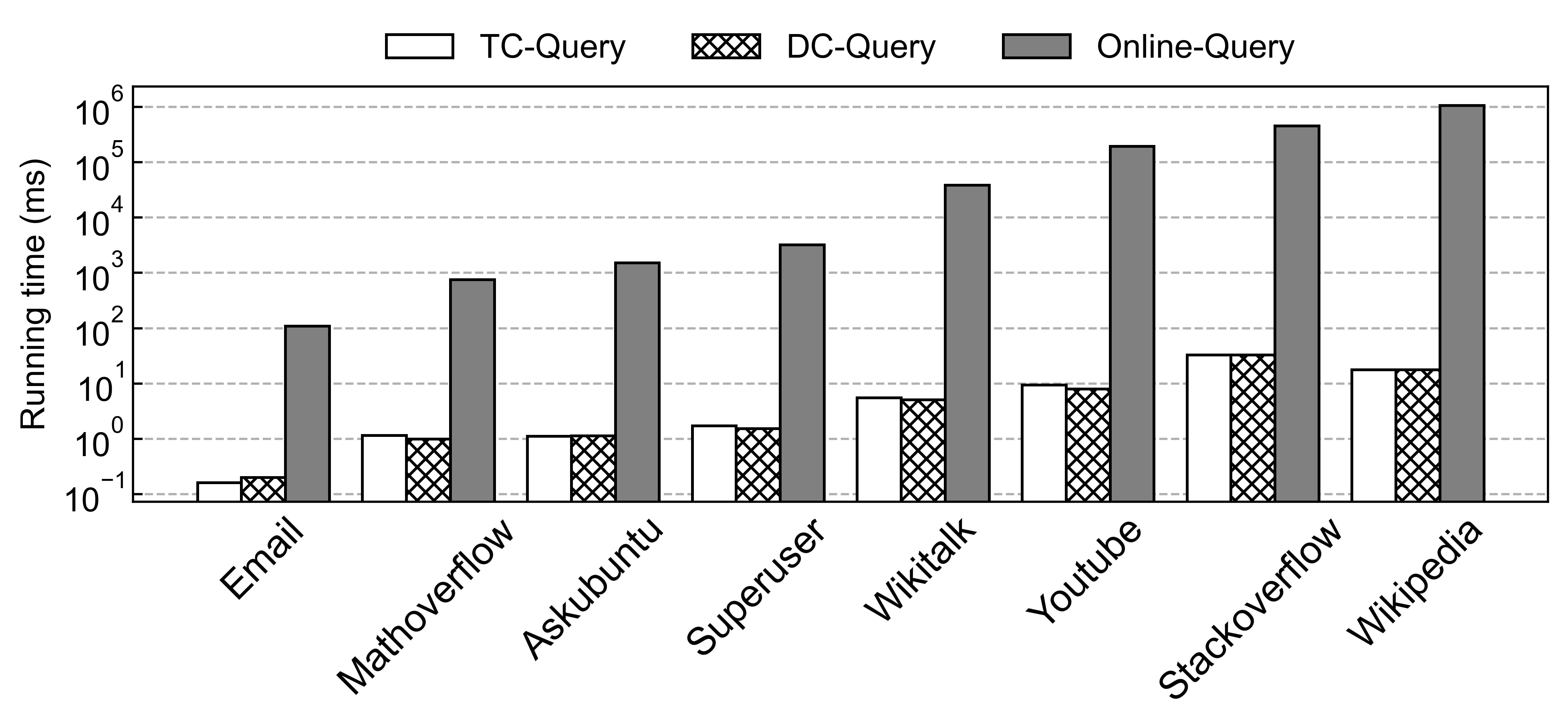}
    \caption{Response time of query processing on different datasets, in which $k = 30\% \cdot  k_{max}$ and $\delta = 60\% \cdot  \delta_{max}$.}
    \label{fig:queryoverallcp}
\end{figure}

\begin{figure*}
    \begin{minipage}[c]{0.73\linewidth}
        \centering
        \captionof{table}{Statistics of datasets.}\label{tab:dataset}
\begin{tabular}{|l|r|r|r|r|r|r|r|}
\hline
\text{Dataset} & $|\mathcal{V}|$ & $|\mathcal{E}|$ & $n$  & $\overline{|\tau|}$ & $|\Delta|$ & $k_{max}$ & $\delta_{max}$\\
\hline\hline
\texttt{Email} & 0.9K & 16K & 803 & 11.5 & 105K & 23 & 800  \\
\hline
\texttt{Mathoverflow} & 24K & 187K & 2450 & 1.6 & 1.4M & 42 & 2336 \\
\hline
\texttt{Askubuntu} & 159K & 455K & 2613 & 1.2 & 680K & 26 & 2040 \\
\hline
\texttt{Superuser} & 194K & 714K & 2773 &  1.2 & 1.5M & 35 & 2692 \\
\hline
\texttt{Wikitalk} & 1.1M & 2.7M & 2320 &  1.4 & 8.1M & 49 & 2231 \\
\hline
\texttt{Youtube} &322K  & 9.3M & 225 &1.0 & 12M & 33 & 225 \\
\hline
\texttt{Stackoverflow} & 2.6M & 28.1M & 2774 &  1.2 & 114.2M & 79  & 2768 \\
\hline
\texttt{Wikipedia} & 1.8M & 36.5M & 2235 & 1.1 & 126.6M & 59 & 2231 \\
\hline
\end{tabular}\vspace{0.1cm}
        \includegraphics[width=0.9 \linewidth]{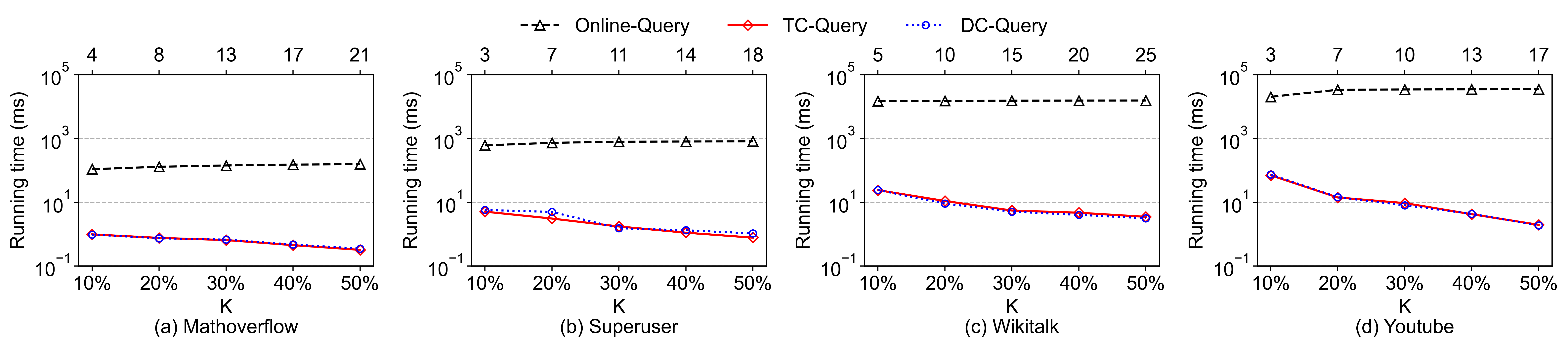}
        \captionof{figure}{Response time of query processing with varying $k$ and $\delta = 60\% \delta_{max}$.}\label{fig:queryvarykcp}
        \includegraphics[width= 0.9\linewidth]{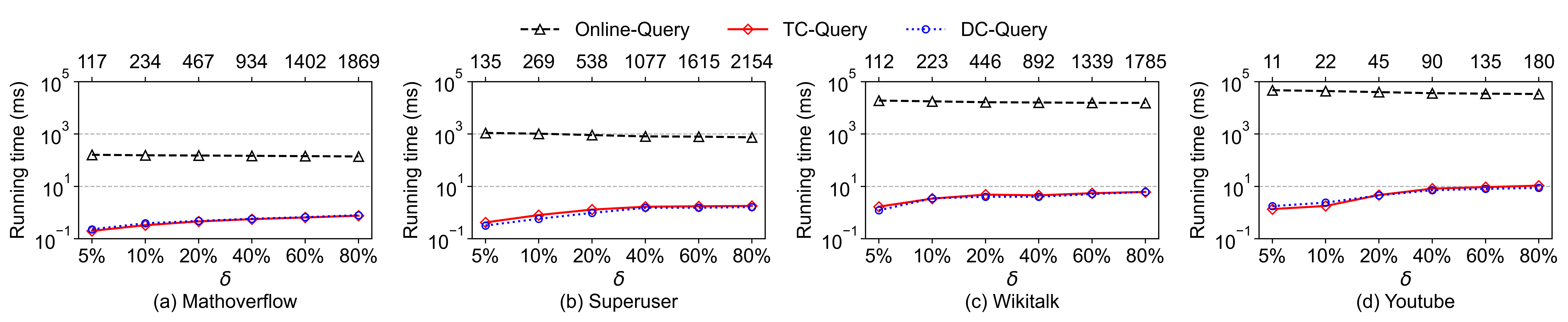}
        \captionof{figure}{Response time of query processing with varying $\delta$ and $k = 30\% k_{max}$.}\label{fig:queryvarydeltacp}
    \end{minipage}\hfill
    \begin{minipage}[c]{0.24\linewidth}
        \centering
        \includegraphics[width=0.9\linewidth]{warm.png}\label{fig:warmwikipedia}
        \captionof{figure}{Response time of DC-Query with varying $k$ and $\delta$.}\label{fig:heat}
    \end{minipage}
    \vspace{-0.5cm}
\end{figure*}

\subsection{Query Processing}

To the best of our knowledge, no existing work investigates the $(k,\delta)$-truss query. Thus, we use the proposed query processing algorithms, Online-Query, TC-Query, and DC-Query to evaluate efficiency. Given $k$ and $\delta$, we record the average running time of 100 times of repeated execution.

Since each dataset has different value ranges of $k$ and $\delta$, we adopt relative parameters and set $k = 30\% k_{max}$ and $\delta = 60\% \delta_{max}$ respectively in default for a given dataset. The running time of three algorithms under the default parameters is reported in Fig~\ref{fig:queryoverallcp}. TC-Query and DC-Query have similar query efficiency, and are 2$\sim$4 orders of magnitude faster than Online-Query. Even for the two largest graphs with tens of millions of edges, index-based approaches can finish within less than 100 ms.

Then we test the query efficiency with respect to varying parameters. Firstly, the running time of three algorithms with varying values of $k$ on four datasets is reported in Fig~\ref{fig:queryvarykcp}. As expected, index-based approaches spend less time on processing queries with greater $k$, because fewer edges need to be scanned. In contrast, Online-Query spends more time because truss decomposition needs to peel more edges. Note that, even when $k = 10\% \cdot  k_{max}$ (which is usually the minimum value 3), the running time of TC-Query and DC-Query is never greater than 100 ms. Moreover, the running time of three algorithms with varying values of $\delta$ on those datasets is reported in Fig~\ref{fig:queryvarydeltacp}. Different from $k$, the running time of index-based approaches increases gradually with increasing $\delta$, because the greater $\delta$ relaxes the constraint and more edges need to be scanned. Lastly, we use heat map to visualize the running time of DC-Query under more combinations of query parameters ($k$ = 10\%, 20\%, $\cdots$, 100\% of $k_{max}$ and $\delta$ = 10\%, 20\%, $\cdots$, 100\% of $\delta_{max}$) in Fig~\ref{fig:heat}. The time cost is at most 0.41 sec on the largest dataset Wikipedia, and generally decreases as $k$ increases or $\delta$ decreases.

% \begin{table*}[t!]
% \centering
% \caption{Statistics of indexes.}
% \label{tab:index}
% \begin{tabular}{|l|r|r|r|r|r|r|}
% \hline
% \multirow{2}{*}{dataset} & avg. entry & TC-Index & \multicolumn{4}{c|}{DC-Index} 
% \\\cline{3-7} 
% & ($k$-span) \# & total edge \# & total edge \# & total edge \# / $|\mathcal{E}|$ & space (MB) & compression ratio (total edge \# / $\sum_{k,\delta}|T_{k,\delta}|$)\\
% \hline\hline
% \texttt{Email} & 290 & 162K & 154K & 9.57 & 0.76 & $17.5\times10^{-4}$\\
% \hline
% \texttt{MathOverflow} & 1365 & 1959K & 1817K & 10.40 & 9.35 & $6.27\times10^{-4}$ \\
% \hline
% \texttt{AskUbuntu} & 1086 & 959K & 958K & 2.10 & 7.33 & $11.43\times10^{-4}$ \\
% \hline
% \texttt{Superuser} & 1365 & 2108K & 2106K & 2.95 & 13.83 & $7.4\times10^{-4}$\\
% \hline
% \texttt{WikiTalk} & 1089 & 10.60M & 10.58M & 3.79 & 62.01 & $7.67\times10^{-4}$\\
% \hline
% \texttt{YouTube} & 170 & 16.74M & 14.25M & 1.52 & 125.93 & $1.11\times10^{-2}$\\
% \hline
% \texttt{Stackoverflow} & 2028 & 139.07M & 138.92M & 4.93 & 746.15 & $6.00\times10^{-4}$\\
% \hline
% \texttt{Wikipedia} & 1304 & 164.24M & 163.40M & 4.47 & 902.63 & $8.60\times10^{-4}$ \\
% \hline
% \end{tabular}
% \end{table*}

\begin{table*}[t!]
\centering
\caption{Statistics of indexes.}
\label{tab:index}
\resizebox{0.9\textwidth}{!}{%
\begin{tabular}{|l|r|r|r|r|r|r|}
\hline
\multirow{2}{*}{dataset} & avg. entry & TC-Index & \multicolumn{4}{c|}{DC-Index} 
\\\cline{3-7} 
& ($k$-span) \# & total edge \# & total edge \# & total edge \# / $|\mathcal{E}|$ & space (MB) & compression ratio (total edge \# / $\sum_{k,\delta}|T_{k,\delta}|$)\\
\hline\hline
\texttt{Email} & 290 & 162K & 154K & 9.57 & 0.76 & $17.5\times10^{-4}$\\
\hline
\texttt{MathOverflow} & 1478 & 1959K & 1871K & 10.40 & 9.35 & $6.25\times10^{-4}$ \\
\hline
\texttt{AskUbuntu} & 1086 & 959K & 958K & 2.10 & 7.33 & $11.43\times10^{-4}$ \\
\hline
\texttt{Superuser} & 1365 & 2108K & 2106K & 2.95 & 13.83 & $7.4\times10^{-4}$\\
\hline
\texttt{WikiTalk} & 1089 & 10.60M & 10.58M & 3.79 & 62.01 & $7.67\times10^{-4}$\\
\hline
\texttt{YouTube} & 170 & 16.74M & 14.25M & 1.52 & 125.93 & $1.11\times10^{-2}$\\
\hline
\texttt{Stackoverflow} & 2028 & 139.07M & 138.92M & 4.93 & 746.15 & $6.00\times10^{-4}$\\
\hline
\texttt{Wikipedia} & 1304 & 164.24M & 163.40M & 4.47 & 902.63 & $8.60\times10^{-4}$ \\
\hline
\end{tabular}
}
\end{table*}

% \begin{table*}[t!]
% \centering
% \caption{Statistics of indexes.}
% \label{tab:index}
% \resizebox{0.8\textwidth}{!}{%
% \begin{tabular}{|l|r|r|r|r|r|r|r|}
% \hline
% \multirow{2}{*}{\textbf{Dataset}} 
% & \multicolumn{2}{c|}{\textbf{TC-Index}} 
% & \multicolumn{5}{c|}{\textbf{DC-Index}} \\
% \cline{2-8}
% & Avg $k$-span \# 
% & Total edge \# 
% & Avg $k$-span \# 
% & Total edge \# 
% & Total edge \#/ $|\mathcal{E}|$ 
% & Space (MB) 
% & Compression ratio ($\sum_{k,\delta}|T_{k,\delta}|$) \\
% \hline
% \texttt{Email} & 290 & 162K & 291 & 154K & 9.57 & 0.76 & $17.5\times10^{-4}$\\
% \hline
% \texttt{Mathoverflow} & 1478 & 1959K & 1479 & 1817K & 10.40 & 9.35 & $6.25\times10^{-4}$ \\
% \hline
% \texttt{Askubuntu} & 1086 & 959K & 1087 & 958K & 2.10 & 7.33 & $11.43\times10^{-4}$ \\
% \hline
% \texttt{Superuser} & 1365 & 2108K & 1366 & 2106K & 2.95 & 13.83 & $7.4\times10^{-4}$\\
% \hline
% \texttt{Wikitalk} & 1089 & 10.60M & 1090 & 10.58M & 3.79 & 62.01 & $7.67\times10^{-4}$\\
% \hline
% \texttt{Youtube} & 170 & 16.74M & 170 & 14.25M & 1.52 & 125.93 & $1.11\times10^{-2}$\\
% \hline
% \texttt{Stackoverflow} & 2028 & 139.07M & 2029 & 138.92M & 4.93 & 746.15 & $6.00\times10^{-4}$\\
% \hline
% \texttt{Wikipedia} & 1304 & 164.24M & 1305 & 163.40M & 4.47 & 902.63 & $8.60\times10^{-4}$ \\
% \hline
% \end{tabular}%
% }
% \end{table*}

\subsection{Index Construction}

Firstly, we report the construction time of proposed TC-Index and DC-Index. For each dataset, we use DBA  to construct TC-Index and MBA to construct both TC-Index and DC-Index respectively. As illustrated in Fig~\ref{fig:IndexConstruct}, index construction costs less than 1 sec for the smallest dataset Email with 16K edges and nearly 3000 sec for the largest dataset Wikipedia with 36M edges. The construction time increases evenly with the increasing scale of graph, which implies that DBA and MBA reduce the redundant computation effectively. Moreover, as expected, MBA is more efficient than DBA on all datasets.

Moreover, the statistics of constructed indexes are shown in Table~\ref{tab:index}. The total edge number of DC-Index is about 1.5X-10.4X large as the corresponding graph. Compared with storing all possible $(k,\delta)$-trusses directly, DC-Index achieves the compression ratio up to $10^{-4}$ on most datasets. The only exceptional dataset is Youtube. According to Theorem~\ref{the:tcspace}, $\delta_{max}$ is the key factor affecting the compression ratio. Since YouTube has a very small $\delta_{max}$, the compression ratio of its DC-Index is worse than other datasets. Due to the effective compression, the space cost of DC-Index is at most 903MB in our experiments. However, we observed that DC-Index is only a little smaller than TC-Index. The rationale is that, since $\delta_{max}$ is much greater than $k_{max}$ for these datasets, the weights of horizontal edges in $E_h$ are generally less than vertical edges in $E_v$, so that the $(k,\delta)$-truss arborescence has only a few vertical edges. As a result, the space of DC-Index is almost as large as TC-Index.

For the above anomaly observed, we conduct an extra experiment to verify the effectiveness of DC-Index. Specifically, we merge every 20, 25, 30, 35, or 40 consecutive timestamps into a single new timestamp for coarsening the time granularity (like from day to month), which will decrease $\delta_{max}$ but not change $k_{max}$. Thus, the weights of horizontal edges in new $E_h$ will become greater. Fig~\ref{fig:granularity} illustrates the comparison of total edge number between DC-Index and TC-Index with the new settings, on two selected datasets. Clearly, DC-Index regains the advantage in the more balanced temporal graphs.

% \begin{figure}[t!]
%     \centering
%     \subfloat[Mathoverflow.]{
%         \includegraphics[width=0.4\linewidth]{rebuild_compare_im.png}
%     }
%     \hfill
%     \subfloat[Superuser.]{
%         \includegraphics[width=0.4\linewidth]{im_tc_and_dc_compare.png}
%     }
%     \caption{Distribution of triangle counts with respect to minimum time span on two datasets.}
%     \label{fig:estudy}
% \end{figure}

% \begin{figure}[t!]
%     \centering
%     \includegraphics[width=0.8\linewidth]{tim_distribution.png}
%     \caption{Update time distribution of $(k,\delta)$-truss index.}
%     \label{fig:indexmaincost}
% \end{figure}
{\color{black}
\subsection{Index Maintenance}
%\noindent\textbf{Update time of index on different datasets}. 
    We randomly remove 1000 edges from Mathoverflow, Askubuntu, Superuser, and Wikitalk and reinsert them into the original graphs, for evaluating the time cost of index maintenance. Figure~\ref{fig:indexmaincost} (a) compares the average update time per insertion of TC-Index (TC-IM) and DC-Index (DC-IM) based on Algorithm~\ref{alg:insertmaintenance} against the index reconstruction from scratch based on MBA. Both TC-IM and DC-IM can achieve up to two orders of magnitude speedup compared to the baseline. For example, on WikiTalk, TC-IM processes edge insertions in 0.36s and DC-IM in 0.53s, while index reconstruction takes 107s and 111s, respectively. Moreover, TC-Index slightly outperforms DC-Index, primarily due to the more complex tree-based structure of DC-Index, which requires additional structural adjustments.

Figure~\ref{fig:indexmaincost} (b) illustrates the distribution of index update time costs for 1,000 edge insertions. For Askubuntu, Superuser, and Wikitalk, a half of the updates are finished within 0.01ms, demonstrating that our proposed filter of $k$ values, $k$-span values, and potentially affected edges are effective. In contrast, for Mathoverflow with much higher clustering coefficient (see Table~\ref{tab:dataset}), the filter of potentially affected edges is not as effective as the other datasets. Additionally, the update time cost distributions of DC-Index and TC-Index are consistent, because the main time costs of both are on Algorithm~\ref{alg:insertmaintenance}. 

}
\begin{figure}[t!]
    \centering
    \includegraphics[width=0.8\linewidth]{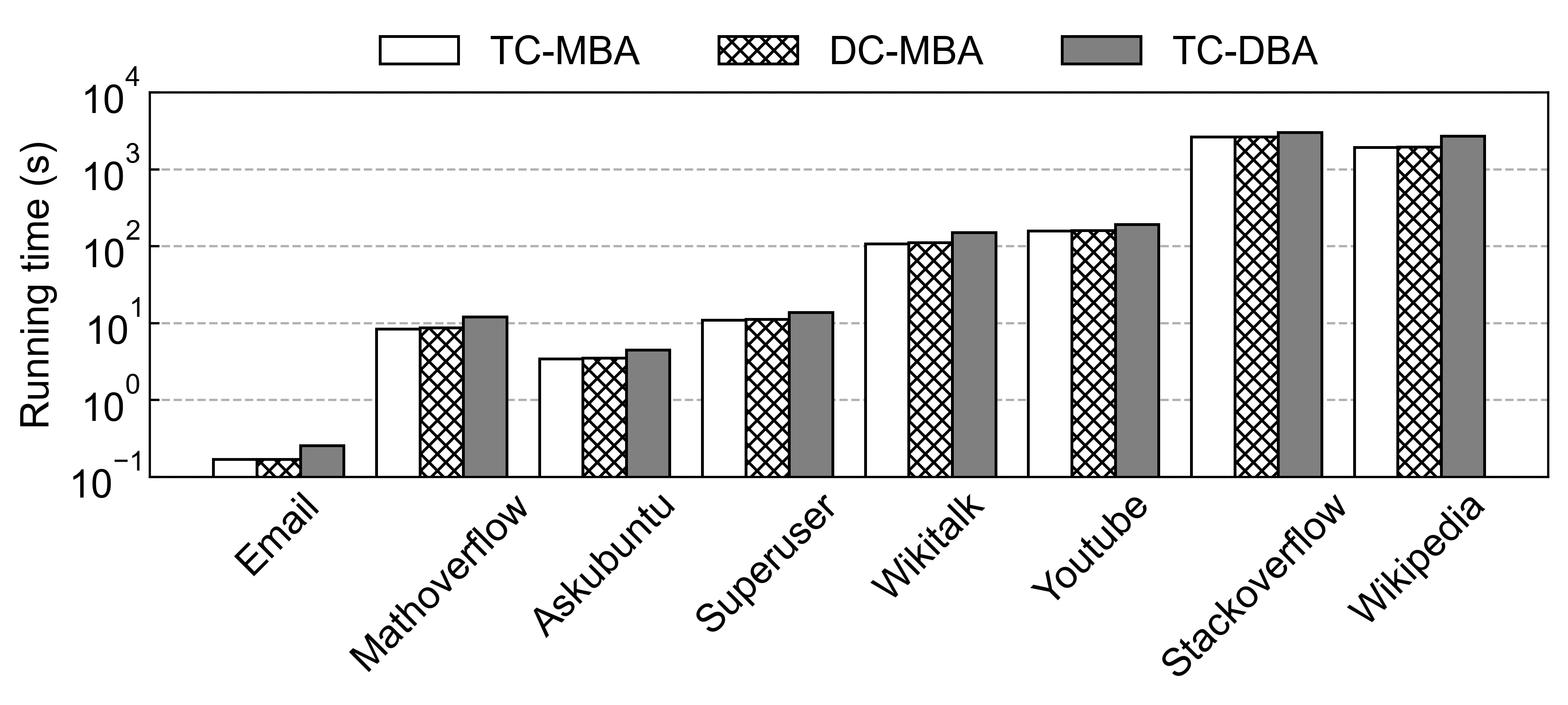}
    \caption{Construction time of TC-Index and DC-Index using DBA and MBA for different datasets.}
    \label{fig:IndexConstruct}
\end{figure}

\begin{figure}[t!]
    \centering
    \includegraphics[width=0.8\linewidth]{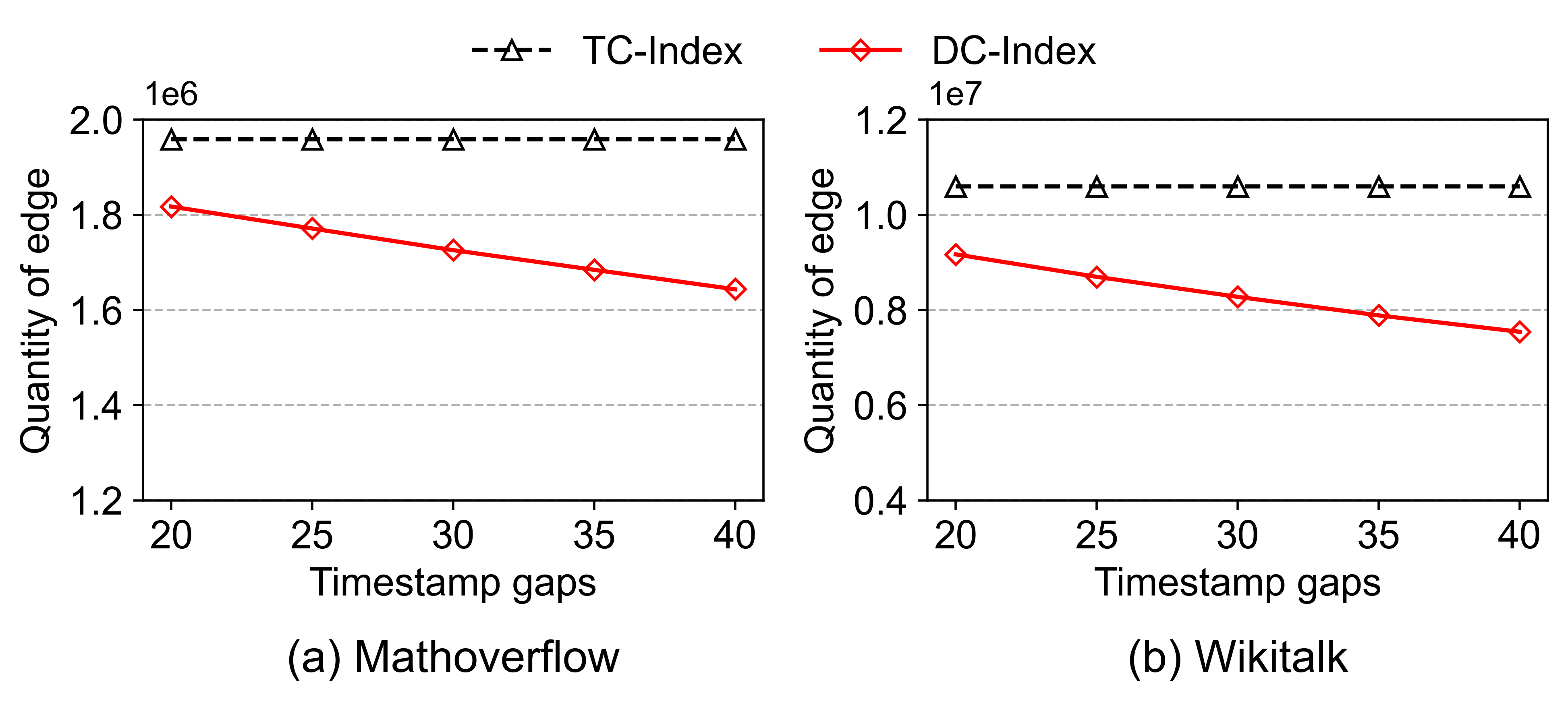}
    \caption{Comparison of total edge number between DC-Index and TC-Index with different time granularity.}
    \label{fig:granularity}
\end{figure}

\begin{figure}[t!]
    \centering
    \includegraphics[width=0.8\linewidth]{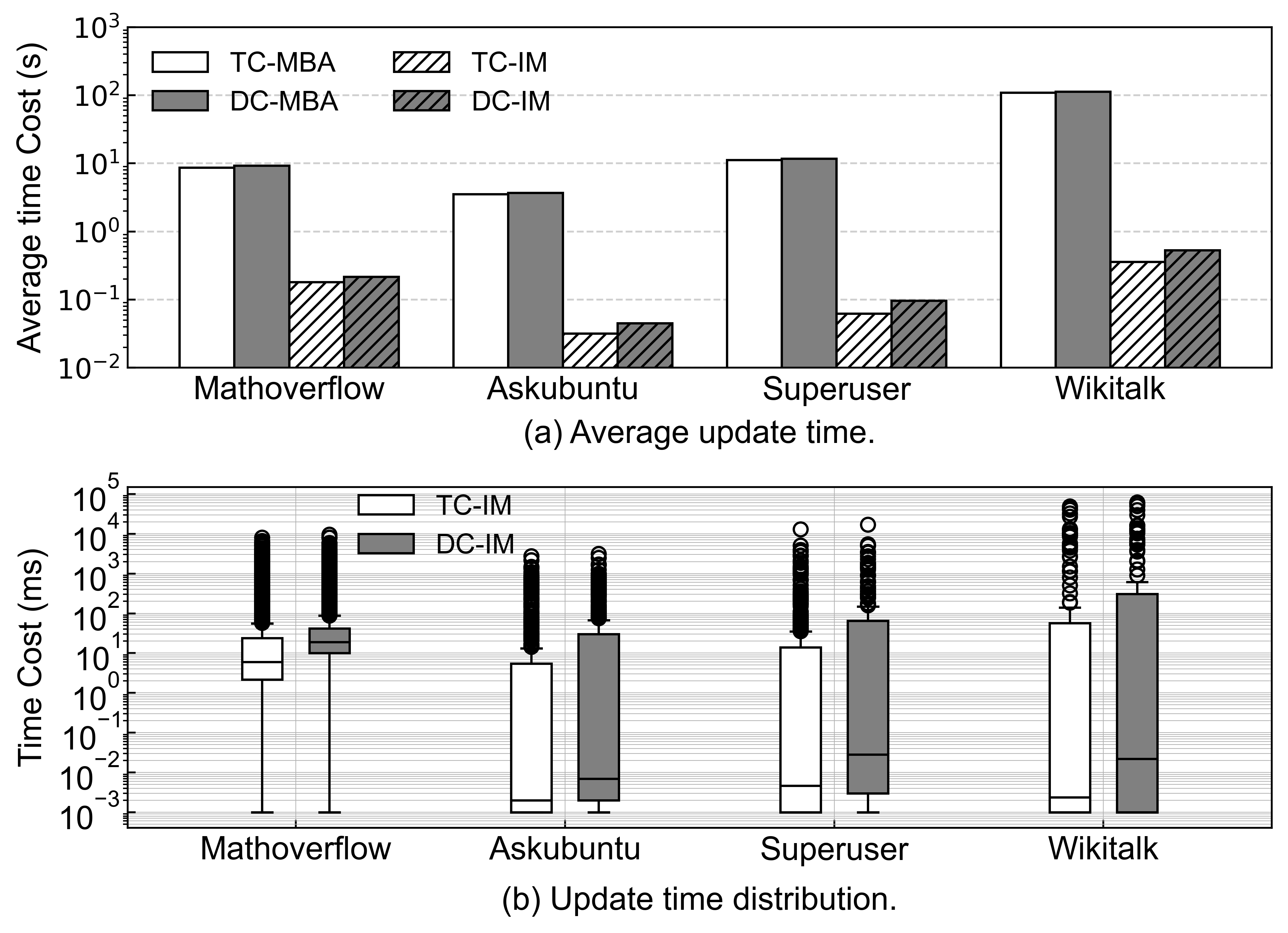}
    \caption{\textcolor{black}{Update time evaluation of $(k,\delta)$-truss index.}}
    \label{fig:indexmaincost}
    \vspace{-0.5cm}
\end{figure}

\section{Related Work}\label{sec:rw}

Recently, there emerge studies on $k$-truss for temporal graphs. The $(k, \Delta, \theta)$-truss~\cite{lantian2020kcommunity} inherits the definition of persistent core~\cite{lirh2018persistcommunity}, and requires its edges to have supports no less than $k-2$ in a number of time windows no shorter than $\Delta$, the total duration of which is no less than $\theta$. Such a truss model is not general enough for many application scenarios. The $(k,\Delta)$-truss (also known as span truss)~\cite{lotito2020spantruss} defines the temporal support of edges, which is the support in a projection of temporal graph during a period, and aims to find the $k$-truss that is cohesive enough in the specific time window $\Delta$. This truss model is more general. However, it always needs the user to give a specific time window in which triangles occur. In contrast, our $(k,\delta)$-truss allows triangles to occur in different time windows that are still short enough. Thus, our $(k,\delta)$-truss is more relaxed than span truss, and meanwhile, can be equivalent to span truss when an extra time window is specified and $\delta = \infty$ if necessary.

Since our $(k,\delta)$-truss is designed on top of temporal triangles, we also investigate the related researches~\cite{kovan2011motif, paran2017motif, liu2019motif, wang2020motif, phasa2021triangle, liu2023motif}. Different from their typical definitions of triangle/motif duration, we propose another kind of duration, namely, minimum time span that is more meaningful in the context of $k$-truss.

Moreover, the design of our indexes is inspired by $(k,l)$-core~\cite{chen2022klcore} and $(k,p)$-core~\cite{zhang2020kpcore}, which also consider the similar property with respect to dual parameters. The various elegant algorithms~\cite{wang2012trussdecomp}~\cite{huang2014trusscommunity, wang2019trussmaint, cliu2014trussmaint, luo2020trussmaintenance} to solve $k$-truss problems provide useful guides to develop our index construction algorithms.

\section{Conclusion}\label{sec:conc}

In this paper, we study a novel $(k,\delta)$-truss on temporal graphs, which constrains the minimum time span of triangles to guarantee temporal cohesion. Such a constraint tailors the static truss in time dimension effectively. To address the query problem of $(k,\delta)$-truss, we propose both index-free and index-based approaches. The indexes that exploits the dual containment property of $(k,\delta)$-truss to compress the space can deliver efficient query processing. Moreover, we develop scalable index construction algorithms \textcolor{black}{and dynamic index maintenance algorithm}. The theoretical proof and experimental evaluation of our approach are provided.

%\section*{Acknowledgment}

%This work was supported by the grants of the National Natural Science Foundation of China (No. 61202036, 61672389, 62272353, and 62276193), the Guangzhou Key Laboratory of Big Data and Intelligent Education (No. 201905010009), and the Research Grants Council of Hong Kong, China (No. 14205520).


\begin{thebibliography}{00}
\bibitem{masuda2016temporalnetwork} N. Masuda, R. Lambiotte, ``A Guide to Temporal Networks,'' World Scientific Publishing Europe Ltd, 2016.
\bibitem{petter2012temporalnetwork} P. Holme, J. SaramakiJari, ``Temporal network,'' Springer, 2012.
\bibitem{kossi2006socialnetwork} G. Kossinets, D. Watts, ``Empirical analysis of an evolving social network,'' science, 311(5757): pp. 88-90, 2006.
\bibitem{huang2022transactionnetwork} X. Huang, Y. Yang, Y. Wang, et al, ``Dgraph: A large-scale financial dataset for graph anomaly detection,'' Advances in Neural Information Processing Systems, 35: pp. 22765-22777, 2022.
\bibitem{filip2023transponetwork} M. Filipovska, H. Mahmassani, ``Spatio-Temporal Characterization of Stochastic Dynamic Transportation Networks,'' IEEE Transactions on Intelligent Transportation Systems, 24(9): pp. 9929-9939, 2023.
\bibitem{hidal2008comnetwork} C. Hidalgo, C. Rodríguez-Sickert, ``The dynamics of a mobile phone network,'' Physica A: Statistical Mechanics and its Applications, 387(12): pp. 3017-3024, 2008.
\bibitem{simeu2021powernetwork} J. Simeunović, B. Schubnel, et al, ``Spatio-temporal graph neural networks for multi-site PV power forecasting,'' IEEE Transactions on Sustainable Energy, 13(2): pp. 1210-1220, 2021.
\bibitem{masuda2014diseasenetwork} N. Masuda, P. Holme, ``Introduction to temporal network epidemiology,'' Springer Singapore, 2017.
\bibitem{galim2018spancore} E. Galimberti, A. Barrat, F. Bonchi, et al, ``Mining (maximal) span-cores from temporal networks,'' Proceedings of the 27th ACM international Conference on Information and Knowledge Management, pp. 107-116, 2018. 
\bibitem{wu2015temporalcore} H. Wu, J. Cheng, et al, ``Core decomposition in large temporal graphs,'' 2015 IEEE International Conference on Big Data (Big Data), pp. 649–658, 2015.
\bibitem{Yang2023temporalcore} J. Yang, M. Zhong, Y. Zhu, T. Qian, M. Liu, and J. Yu, ``Scalable time-range k-core query on temporal graphs,'' PVLDB, 16(5): pp. 1168–1180, 2023.
\bibitem{myu2021historicalcore} M. Yu, D. Wen, L. Qin, et al, ``On querying historical k-cores,'' PVLDB, 14(11): pp. 2033–2045, 2021.
\bibitem{Wang2025historicalcore} Z. Wang, M. Zhong, Y. Zhu, et al. ``On More Efficiently and Versatilely Querying Historical k-Cores, `` PVLDB, 18(5): pp. 1335-1347, 2025.
\bibitem{zhong2024kxcore}M. Zhong, J. Yang, Y. Zhu, T. Qian, M. Liu, J. Yu, ``A Unified and Scalable Algorithm Framework of User-Defined Temporal $(k,\mathcal{X})$-Core Query,'' IEEE Transactions on Knowledge and Data Engineering, 2024.
\bibitem{Bai2020frequentcore} W. Bai, Y. Chen, and D. Wu, ``Efficient temporal core maintenance of massive graphs,'' Information Sciences, 513: pp. 324–340, 2020.
\bibitem{Du2025frequentcore} Z. Du, M. Zhong, Y. Zhu, et al. ``Efficient Frequency-Aware k-Core Query on Temporal Graphs, `` ICDE, pp. 2366-2379, 2025.
\bibitem{lirh2018persistcommunity} R. Li, J. Su, L. Qin, J. Yu, and Q. Dai, ``Persistent community search in temporal networks,'' ICDE, pp. 797–808, 2018.
\bibitem{chu2019burstingsubgraph}  L. Chu, Y. Zhang, Y. Yang, L. Wang, and J. Pei, ``Online density bursting subgraph detection from temporal graphs,'' PVLDB, 12(13): pp. 2353–2365, 2019.
\bibitem{qin20202periodiccore}  H. Qin, R. Li, Y. Yuan, G. Wang, W. Yang, and L. Qin, ``Periodic communities mining in temporal networks: Concepts and algorithms,'' IEEE Transactions on Knowledge and Data Engineering, 34(8): pp. 3927-3945, 2020.
\bibitem{liliu2021continuecore} Y. Li, J. Liu, H. Zhao, J. Sun, Y. Zhao, and G. Wang, ``Efficient continual cohesive subgraph search in large temporal graphs,'' World Wide Web, 24(5): pp. 1483–1509, 2021.
\bibitem{tang2022reliablecore} Y. Tang, J. Li, N. Haldar, Z. Guan, J. Xu, and C. Liu, ``Reliable community search in dynamic networks,'' PVLDB, 15(11): pp. 2826–2838, 2022.
\bibitem{lotito2020spantruss} Q. Lotito, A. Montresor, ``Efficient Algorithms to Mine Maximal Span-Trusses From Temporal Graphs,'' unpublished.
\bibitem{lantian2020kcommunity} L. Xu, R. Li, G. Wang, et al, ``Research on K-truss Community Search Algorithm for Temporal Networks,'' Journal of Frontiers of Computer Science and Technology, 14(9): pp. 1482-1489, 2020.
\bibitem{phasa2021triangle} N. Pashanasangi, and C. Seshadhri, ``Faster and Generalized Temporal Triangle Counting, via Degeneracy Ordering,'' KDD, pp. 1319--1328, 2021.
\bibitem{wang2020motif} J. Wang, Y. Wang, W. Jiang, Y. Li, and K. Tan, ``Efficient Sampling Algorithms for Approximate Temporal Motif Counting,'' CIKM, pp. 1505–1514, 2020.
\bibitem{paran2017motif} A. Paranjape, A. Benson, and J. Leskovec, ``Motifs in Temporal Networks,'' WSDM, pp. 601–610, 2017.
\bibitem{liu2019motif} P. Liu, A. Benson, and M. Charikar, ``Sampling Methods for Counting Temporal Motifs,'' WSDM, pp. 294–302, 2019.
\bibitem{liu2023motif} P. Liu, V. Guarrasi, and A. Sarıyuce, ``Temporal Network Motifs: Models, Limitations, Evaluation,'' IEEE Transactions on Knowledge and Data Engineering, 35: pp. 945--957, 2023.
\bibitem{kovan2011motif} L. Kovanen, M. Karsai, K. Kaski, J. Kertész, and J. Saramäki, ``Temporal motifs in time-dependent networks,`` Journal of Statistical Mechanics: Theory and Experiment, P11005, 2011.
\bibitem{leskovec2005graphanalysis} J. Leskovec, J. Kleinberg, C. Faloutsos, ``Graphs evolution: Densification and shrinking diameters'' ACM transactions on Knowledge Discovery from Data (TKDD), 1(1): pp. 2–42, 2007.
\bibitem{wang2012trussdecomp} J. Wang, J. Cheng, ``Truss decomposition in massive networks,'' PVLDB, 5(9): pp. 812–823, 2012.
\bibitem{chen2014trussdecomp} P. Chen, C. Chou, and M. Chen, ``Distributed algorithms for k-truss decomposition,'' IEEE International Conference on Big Data (Big Data), pp. 471–480, 2014.
\bibitem{che2010trussdecomp} Y. Che, Z. Lai, S. Sun, Y. Wang, and Q. Luo, ``Accelerating truss decomposition on heterogeneous processors,'' Proceedings of the VLDB Endowment, 13(10): pp. 1751-1764, 2020.
\bibitem{kabir2017paratrussdecomp} H. Kabir, K. Madduri, ``Parallel k-truss decomposition on multicore systems,'' IEEE High Performance Extreme Computing Conference (HPEC), pp. 1-7, 2017.
\bibitem{kabir2017sharedtrussdecomp} H. Kabir, K. Madduri, ``Shared-memory graph truss decomposition,'' IEEE 24th International Conference on High Performance Computing (HiPC), pp. 13–22, 2017.
\bibitem{huang2014trusscommunity} X. Huang, H. Cheng, L. Qin, et al, ``Querying k-truss community in large and dynamic graphs,'' Proceedings of the 2014 ACM SIGMOD international conference on Management of data, pp. 1311-1322, 2014.
\bibitem{wang2019trussmaint} Y. Zhang, J. Yu, ``Unboundedness and efficiency of truss maintenance in evolving graphs,'' Proceedings of the 2019 International Conference on Management of Data, pp. 1024-1041, 2019.
\bibitem{cliu2014trussmaint} R. Zhou, C. Liu, J. Yu, et al, ``Efficient truss maintenance in evolving networks,'' arXiv preprint arXiv:1402.2807, 2014.
\bibitem{luo2020trussmaintenance} Q. Luo, D. Yu, X. Cheng, et al, ``Batch processing for truss maintenance in large dynamic graphs,'' IEEE Transactions on Computational Social Systems, pp. 1435-1446, 2020.
\bibitem{suntrussmaintenance} Z. Sun, X. Huang, Q. Liu, et al, ``Efficient Star-based Truss Maintenance on Dynamic Graphs,'' Proceedings of the ACM on Management of Data, 1(2): pp. 1-26, 2023.
\bibitem{batage2003coredecomp} V. Batagelj and M. Zaversnik, ``An o (m) algorithm for cores decomposition of networks,'' arXiv preprint cs/0310049, 2003.
\bibitem{lesko2014snap}  J. Leskovec, A. Krevl, ``SNAP Datasets: Stanford large network dataset collection,'' http://snap.stanford.edu/data, Jun. 2014.
\bibitem{kuneg2013konect} J. Kunegis, ``Konect: the koblenz network collection,'' in Proceedings of the 22nd international conference on world wide web, pp. 1343–1350, 2013.
\bibitem{chen2022klcore} Y. Chen, J. Zhang, Y. Fang, et al, ``Efficient community search over large directed graphs: An augmented index-based approach,'' Proceedings of the Twenty-Ninth International Conference on International Joint Conferences on Artificial Intelligence, pp. 3544-3550, 2021.
\bibitem{zhang2020kpcore} C. Zhang, F. Zhang, W. Zhang, et al, ``Exploring finer granularity within the cores: Efficient (k, p)-core computation,'' ICDE, pp. 181-192, 2020.
\bibitem{Hu2024kdtruss} C. Hu, M. Zhong, Y. Zhu, et al. ``Querying cohesive subgraph regarding span-constrained triangles on temporal graphs,`` ICDE, pp. 3338-3350, 2024.


\end{thebibliography}
\end{document}